%% file: main.tex
\documentclass[]{IEEEtran}
\IEEEoverridecommandlockouts

\usepackage[utf8]{inputenc}       
\usepackage[T1]{fontenc}          
\usepackage[english]{babel}

\usepackage{amsmath}
\usepackage{amssymb}
\usepackage{amsfonts}
\usepackage{amsthm}               
\usepackage{mathrsfs}
\usepackage{bm}                   
\usepackage{bbm}                  
\usepackage{nicefrac}             
\usepackage{upgreek}

\newcommand{\QED}{\hfill\ensuremath{\blacksquare}}

\theoremstyle{definition} 
\newtheorem{definition}{Definition}
\newtheorem{assumption}{Assumption}
\newtheorem{remark}{Remark}

\theoremstyle{definition} 

\newtheorem{example}{Example}[section]
\theoremstyle{plain}
\newtheorem{theorem}{Theorem}

\usepackage{ifpdf}
\ifpdf
    \usepackage[pdftex]{graphicx}
    \graphicspath{{./figures/}}
\else
    \usepackage[dvips]{graphicx}
    \graphicspath{{./figures/}}
\fi

\usepackage{epstopdf}
\usepackage{caption}
\usepackage{subcaption}           
\usepackage{graphics}
\usepackage[dvipsnames]{xcolor}   

\usepackage{pgf}
\usepackage{tikz}
\usepackage{pgfplots}
\pgfplotsset{compat=1.18}
\usetikzlibrary{shapes, arrows, automata}
\tikzset{every node/.style={font=\LARGE}}

\usepackage{booktabs}             
\usepackage{multirow}
\usepackage{rotating}
\usepackage{cuted}                

\usepackage{algorithm}
\usepackage{algpseudocode}
\usepackage{enumerate}
\usepackage{microtype}             

\usepackage{times}                
\usepackage{cite}                 

\input{extras/mySections}

\input{extras/mySymbol.sty}

\definecolor{mutedblue}{RGB}{70, 130, 180}   
\definecolor{mutedred}{RGB}{178, 34, 34}     
\definecolor{mygreen}{RGB}{34, 139, 34}  

\usepackage[  
    colorlinks = true,       
    pdfborder  = {0 0 0},    
    urlcolor   = mutedblue,  
    linkcolor  = mutedred,   
    citecolor  = mygreen
]{hyperref}

\title{Scalability of Graph Neural Network Policies in Wireless Communication Networks}

\author{Romina Garcia Camargo \quad Zhiyang Wang \quad Alejandro Ribeiro \thanks{RGC and AR are with the Department of Electrical and Systems Engineering, University of Pennsylvania, Philadelphia, PA (emails: \{rominag, aribeiro\}@engineering.upenn.edu). ZW is with Department of Electrical and Systems Engineering, Washington University in St. Louis, St. Louis, MO (email: wzhiyang@wustl.edu). Preliminary results were presented at ICASSP 2026 \cite{garciacamargoscalabilitygnnrgg} and SPAWC 2026 \cite{camargo2026limitanalysisgraphneural}.} }

\newcommand{\rofs}[1]{#1 \odot \big[\, \mathbf{1} - \bbS_c#1    \,\big]_+}

\def \rst {\rofs{\bbp_c(t)}}
\begin{document}
\pagestyle{plain}
\maketitle

\begin{abstract}
Graph Neural Networks (GNNs) offer scalable solutions for wireless resource allocation, yet existing formal performance guarantees across varying scales for spatial and sparse graphs do not directly apply to these settings. Scalability theories rely on dense graphon limits or continuous manifold approximations, both of which fail under the sparse, bounded-degree regimes and Euclidean metric constraints of physical wireless environments. This paper establishes a theoretical framework for GNN transferability over sparse Random Geometric Graphs (RGGs), capturing distance-dependent channel decay and spatial interference. We model sparse RGGs as spatial perturbations of regular Deterministic Grid Graphs (DGGs) and employ spatial windowing operators to compare networks across differing scales. Assuming Lipschitz continuity of GNN architectures and signal stationarity, we prove scalability over DGGs and bound same-scale transferability between DGGs and RGGs. Combining these results establishes formal scalability bounds across sparse RGG topologies. Finally, we extend this framework to conflict graph models, deriving equivalent scalability guarantees for link-level resource allocation policies. We verify our results for scalability with two experiment settings: power allocation and wireless link scheduling. The simulations show GNNs exhibit the expected scalable behavior, and analyze the relevance of our theoretical assumptions in practical deployment.
\end{abstract}


\section{Introduction}
\label{sec:intro}
\input{01_intro}
\section{Scalability of Wireless Policies Parameterized via Graph Neural Networks}
\label{sec:transfwirelesspolicies}
\input{02_wireless_policies_scalability}

\section{Scalability of Graph Neural Networks on Random Geometric Graphs}
\label{sec:transfgnns}
\input{03_transferability_gnns}

\section{Scalability of Graph Neural Networks on Conflict Graphs}
\label{sec:transfgnnsconflict}
\input{04_transferability_gnns_conflicts}

\section{Numerical Experiments}
\label{sec:numericalexperiments}
\input{05_numerical_experiments}

\section{Conclusions}
\label{sec:conclusion}
\input{conclusions}

\urlstyle{same}
\bibliographystyle{IEEEtran}
\bibliography{references}
\appendix
\section{Supplementary Material}
\label{sec:appendix}
\input{appendix}

\end{document}

%% file: extras/mySections.tex
\usepackage{needspace}



%% file: 01_intro.tex
Learning-based solutions provide mechanisms for efficient resource allocation in wireless networks. Given the structural mapping of devices and their communication links to nodes and edges, current approaches employ Graph Neural Networks (GNNs) \cite{uslu26faststateaugmented, wang2022learningdecentralized, zhou25wirelessdynamic, camargo2026longhorizonwirelesslinkscheduling}. GNNs comprise a cascade of graph filters and pointwise nonlinearities with learnable parameters independent of the underlying graph dimension. This independence enables model deployment on unseen communication topologies, including instances with larger device counts.

Scalability—the capacity to scale to larger graphs while preserving performance—is a well-studied property of GNNs \cite{levin2026transferring, ruiz2021transferability, maskey2021transferability, keriven2020convergence, levie2019transferability, wang2024geometric, shu2024transferabilitydownsampedsparsegraph, shu2026size, le2023graphopssparse}. Theoretical frameworks often rely on graph limit objects, such as graphons, to prove GNNs evaluated on larger graphs preserve performance \cite{ruiz2021transferability, maskey2021transferability, keriven2020convergence}. However, graphons inherently model dense networks, which does not comply with the natural sparsity of wireless networks. Alternative frameworks model relatively-sparse graphs as discretizations of continuous topological spaces, allowing scalability when graphs originate from the same manifold \cite{levie2019transferability, wang2024geometric}. Recent approaches downsample dense graph sequences generated from a graphon to produce sparse structures that preserve underlying topological features \cite{shu2024transferabilitydownsampedsparsegraph}. Subsequent work refines this process via generalized graphons and post-sampling stretching to establish expected upper bounds on transfer error across varying graph sizes \cite{shu2026size}. The authors in \cite{le2023graphopssparse} bypass discrete node sampling by sampling limit operators. This approach yields transferability bounds across varying sizes on sparse graphs but restricts explicit topological analysis. These limit representations abstract graph topology away from physical space, omitting the spatial geometric embedding required to model distance-based channel physics and spatial interference. 

Within wireless resource allocation, scalability literature focuses on architectural invariance or continuous domain limits \cite{shen2021scalablerrm, Wu2024OnTS, wang2022stable}. Initial frameworks demonstrate scale generalization by leveraging permutation equivariance and mean-pooling aggregation to prevent output divergence as graph dimensions increase \cite{shen2021scalablerrm, Wu2024OnTS}. However, these architectural properties guarantee structural invariance across matrix sizes rather than bounding performance transfer error over changing spatial domains. To incorporate spatial geometry, continuous manifold approaches evaluate resource allocation stability under relative perturbations to the Laplace-Beltrami operator \cite{wang2022stable}. Yet, because discrete graph operator convergence to a continuous manifold limit requires high node density, these stability guarantees break down under the sparse, bounded-degree topologies characteristic of practical wireless systems.

We study Graph Neural Network scalability over sparse Random Geometric Graphs with a focus on wireless networks (Sections \ref{sec:transfwirelesspolicies}, \ref{sec:transfgnns}). RGGs uniformly sample node positions in Euclidean space and connect nodes within a fixed connectivity radius, capturing the distance-dependent power decay and spatial operational limits of wireless channels. We formulate RGGs as stochastic spatial perturbations of Deterministic Grid Graphs (DGGs), which arrange nodes in regular spatial lattices that model cellular base station deployments.

Our theoretical framework proves GNN scalability across RGGs by establishing topological closeness to underlying DGGs (Assumption \ref{ass:closeness}) and applying spatial windowing operators to compare graphs across scales (Assumptions \ref{ass:windowingoperationDGGs}, \ref{ass:windowingoperationRGGs}). Under Lipschitz graph filters and activation functions alongside signal stationarity (Assumptions \ref{ass:stationarity}, \ref{ass:lipschitznonlinearity}), \ref{ass:lipschitzfilter}, the theoretical pipeline in Figure \ref{fig:rgg-dgg-complete} proves GNN scalability over DGGs (Theorem \ref{th:dggonndgg}) and establishes same-scale transferability between RGGs and DGGs via graph perturbation bounds (Theorem \ref{the:rgg-gg-gnntransf}). Bridging these results yields scalability across RGGs (Theorem \ref{the:rggtonrgg}), providing formal size generalization guarantees for node-based wireless optimization policies.

Because link-level scheduling tasks model mutual interference via conflict graphs \cite{firstconflictgraph, linkschedulingusinggnns, camargo2026longhorizonwirelesslinkscheduling}, we extend this perturbation framework to edge-based formulations. This yields equivalent scalability guarantees for link-level resource allocation policies (Theorems \ref{th:dggonndgg-conflicts}, \ref{the:rgg-gg-gnntransf-conflict}, \ref{the:rggtonrgg-conflict}, Section \ref{sec:transfgnnsconflict}). 

Numerical simulations evaluate GNN scalability for power allocation over sparse RGGs and link scheduling over conflict graphs. We benchmark against heuristic baselines to demonstrate scale generalization, with GNNs outperforming traditional algorithms across scales. Furthermore, the experiments validate core theoretical assumptions by evaluating the impact of topological closeness between DGGs and RGGs and the need of Lipschitz-continuous graph filters.

%% file: 02_wireless_policies_scalability.tex
In a wireless network with $m$ users, $\mathbf{x}_m\in\mathbb{R}^m$ describes each user's state, while the channel states are summarized in a matrix $\mathbf{S}_m\in\mathbb{R}^{m\times m}$ (see Def. \ref{def:channelinterferencegraph} and \ref{def:conflictgraph}). The objective of resource management is to allocate resources $\bbp_m$ according to the user and channel states. After an allocation is determined, the system returns a reward $\bbf$, which measures the instantaneous Quality of Service (QoS). As long horizon QoS is closer to user experience, we take expectations on the joint distribution of states $(\mathbf{x}_m,\mathbf{S}_m)$,
\begin{align}
    \mathbf{r}_m = \mathbb{E}[\mathbf{f}
    (\mathbf{p}_m; \mathbf{x}_m, \mathbf{S}_m)].
    \label{eq:reward}
\end{align}
Resource management maximizes a long-term network-wide utility $u_0(\bbr_m)$, which forms the optimization objective. Moreover, we introduce $\mathbf{u}(\bbr_m)$ to capture long-term system constraints. A general resource allocation task can be formulated as the following constrained optimization problem:
\begin{align}\label{eq:problemform}
 \mathbf{p}_m^\star(\mathbf{x}_m, \mathbf{S}_m) = & \argmax_{\mathbf{p}_m(\mathbf{x}_m,\mathbf{S}_m)} \ 
    u_0(\mathbf{r}_m), 
    \\
    & \st \quad 
  \mathbf{u}(\mathbf{r}_m)  \geq\mathbf{0}.\nonumber
\end{align}

The objective is often to maximize the network sum rate over time, where per-user rates are limited by interference. In wireless networks, devices exchange information over a shared electromagnetic medium, and the resulting interference can be modeled using an interference or a conflict graph. We define these two graphs explicitly in the following.

\begin{definition}(Channel Interference Graph)\label{def:channelinterferencegraph}
    Consider a network of $m$ users with positions $x_1, \ldots, x_m \in \mathbb{R}^2$.
    The channel interference graph is the weighted graph with adjacency matrix $\bbS_m \in \mathbb{R}^{m \times m}$ that collects the channel gains between user pairs. Accounting for path loss and small-scale fading, its entries are given by
    \begin{align}
        [\bbS_m]_{ij} = \frac{g_{ij}}{d(i,j)^{2}}, \quad i \neq j, \qquad [\bbS_m]_{ii} = 0,
    \end{align}
    where $d(i,j) = \|x_i - x_j\|$ is the distance between users $i$ and $j$, and the fading gains $g_{ij} \sim \mathrm{Exp}(1)$ are drawn i.i.d.\ across user pairs. Under channel reciprocity, $g_{ij} = g_{ji}$, so that $\bbS_m$ is symmetric.
\end{definition}

The primary interference model assumes two links will interfere with each other if they share a user and are scheduled simultaneously, due to the hidden node problem. Conflict graphs indicate which links could result in collisions \cite{hajekb88}. 

\begin{definition} (Conflict Graph)\label{def:conflictgraph}
    Let $\ccalL = \{\ell_1, \ldots, \ell_c\}$ be a set of $c$ communication links, where each link $\ell = (i, j)$ is a pair of users. Two distinct links \emph{conflict} if they share a user. The conflict graph is the graph on $\ccalL$ with adjacency matrix $\bbS_c \in \{0,1\}^{c \times c}$ given by
    \begin{align}\label{eq:conflict-graph}
        [\bbS_c]_{\ell, \ell'} = 1 \iff \ell \neq \ell' \text{ and } \{i, j\} \cap \{i', j'\} \neq \emptyset,
    \end{align}
    for links $\ell = (i,j)$ and $\ell' = (i',j')$.
\end{definition}

We follow with two examples of resource management tasks to illustrate the scenarios for each channel state matrix. 

\begin{example}[Power Allocation]
    \label{ex:powerallocation}
Power allocation assigns transmit powers $\bbp_m(\bbx_m, \bbS_m)$ to maximize the network sum rate derived from the Shannon capacity. Assuming additive white Gaussian noise with variance $\eta^2$, the achievable rate for user $i$ is
\begin{equation}\label{eq:rates}
    r_i(\bbx_m, \bbS_m) := \log \Bigg(1+\frac{|s_{ii}|^2 p_i(\bbx_m, \bbS_m)}{\eta^2+\sum_{k \in \ccalN(i)}|s_{ik}|^2 p_k(\bbx_m, \bbS_m)}\Bigg), 
\end{equation}
where $\ccalN(i)$ represents the neighborhood of user $i$, i.e. all users within a radius $r$. There is limited power $P_{\max}$ that can be allocated, which requires that we set
\begin{align}\label{eq:paconstriant}
    \bbu(\bbr_m) = \mathbb{E}[\mathbf{1}^\mathsf{T}\mathbf{p}_m(\mathbf{x}_m, \bbS_m)] \leq P_\text{max}. 
\end{align}

We maximize the expected sum rate over all channel realizations (Eq. \eqref{eq:rates}), subject to a network-wide average power budget (Eq. \eqref{eq:paconstriant}). As in Equation \eqref{eq:problemform}, we seek the optimal allocation $\bbp^\star_m$.
\qed
\end{example}

\begin{example}[Wireless Link Scheduling]
    \label{ex:wls}
The scheduling decisions $\bbp_c(t)\in\{0,1\}^{c}$ dictate which links transmit at time $t$ in a network with $c$ links. The long-term successful transmission rate is given by
\begin{align}\label{eqn_rate_pointwise}
    \bar{\bbr}_c(\bbs) = \frac{1}{T}\sum_t \mathbf{p}_c(t) \big[\, \mathbf{1} - \bbS_c\mathbf{p}_c(t)     \,\big]_+.
\end{align}

Transmissions are successful when links that could collide are not scheduled simultaneously. A minimum transmission requirement $\bbdelta \in \reals_+^{c}$ ensures that link $i$ remains active for at least a fraction of time $\bbdelta_i$:
\begin{align}\label{eq:mintxreq}
    \bbu(\bbr_c) = \frac{1}{T}\sum_t \rst \geq \bbdelta,
\end{align}
where $\odot$ is the dot product to obtain per-link rates. 
We maximize the long-term successful rates while ensuring all links get access to the channel. Long-horizon wireless link scheduling requires $T$ allocation vectors $\bbp^\star_c$ that solve Equation \eqref{eq:problemform}, where the objective is replaced by Equation \eqref{eqn_rate_pointwise} and the constraint is given by Equation \eqref{eq:mintxreq}.
\qed
\end{example}

\subsection{Learning Resource Allocation Policies on Graphs}
The resource allocation policy can be framed as a mapping between graph signals, which are learned with graph convolutional filters. 
Let $\bbh$ denote a filter that diffuses an input signal across the graph, capturing local topology through $\bbS_m$, the Graph Shift Operator (GSO)  \cite{du18ongraphconvforgraphcnns, segarra17optimalgraphfilterdesign, gama19cnnforsignalsongraphs}. The operation is a polynomial expansion of the GSO, where the $k$-th power shifts the signal over the $k$-hop neighborhood of each node. Given a graph signal $\bbx_m$, the output graph signal $\bby_m \in \mathbb{R}^m$ is defined as
\begin{align}\label{graphfilter}
    \bby_m = \bbh(\bbS_m)\bbx_m = \sum_{k=0}^{K} h_k\bbS_m^k\bbx_m,
\end{align}
where $K$ is the filter order, restricting diffusion to a $K$-hop radius. 
To learn non-linear allocation policies, pointwise non-linearities are integrated between filters in a layered architecture, yielding Graph Neural Networks.

GNNs alternate graph convolutional filters with a pointwise non-linearity $\gamma: \mathbb{R} \rightarrow \mathbb{R}$. The output of layer $l$ is recursively defined as
\begin{align}\label{gnn}
    \bbx_m^{(l)} = \gamma\left(\sum_{k=0}^{K} h^{(l)}_{k}\bbS_m^k\bbx_{m}^{(l-1)}\right),
\end{align}
where $\bbx^{(0)} = \bbx$ and $\bbx^{(l-1)}$ is the output from the preceding layer. Denoting the collective trainable parameters across all $L$ layers as $\mathcal{H}_m$, the overall network mapping is expressed as 
\begin{align}\label{eq:gnndefinition}
    \bbPhi(\bbx_m, \bbS_m; \mathcal{H}_m) = \bbx_m^{(L)}.
\end{align}

In resource management, GNNs serve as a proxy for the allocation policy. The architecture is trained to maximize the network utility function, finding optimal parameters $\mathcal{H}_m^\star\in\mathbf{H}$. The reward presented in Equation \eqref{eq:reward} can be parameterized as 
\begin{align}
    \mathbf{r}_m({\bbPhi}) = \mathbb{E}[\mathbf{f}
    (\bbPhi(\bbx_m, \bbS_m; \mathcal{H}_m); \mathbf{x}_m, \mathbf{S}_m)].
    \label{eq:rewardparam}
\end{align}
Equation \eqref{eq:rates} can be reformulated with this parameterization, which results in a learning task \cite{eisen2020optimal}:
\begin{align}\label{eq:problemformparameterized}
 \bbPhi^\star:= \bbPhi(\bbx_m, \bbS_m;\ccalH_m^\star)= & \argmax_{\bbPhi} \ 
    u_0(\mathbf{r}_m({\bbPhi})),
    \\
    & \st \quad 
  \mathbf{u}(\mathbf{r}_m({\bbPhi}))  \geq\mathbf{0}.\nonumber
\end{align}

While Equations \eqref{eq:problemform} and \eqref{eq:problemformparameterized} are analogous, the introduction of a parameterization enables learning solutions to the resource allocation task. This learned policy can be executed without the need for online optimization.

\begin{remark}
    In practice, GNNs utilize multi-channel features to increase expressivity. This extension generalizes the graph filter to a matrix operation:
    \begin{equation}\label{gnnfilter}
        \bbX_m^{(l)} = \gamma\left(\sum_{k=0}^{K} \bbS_m^k\bbX_{m}^{(l-1)}\bbH^{(l)}_{k}\right),
    \end{equation}
    where $\bbH^{(l)}_{k}\in\mathbb{R}^{F_{l-1}\times F_l}$ denotes the filter coefficient matrix for layer $l$ and converts data from $F_{l-1}$ to $F_l$ features. The input and output signals can also have multiple features, as $\bbX_0=\bbx \in \mathbb{R}^{m \times F_0}$ and $\bbPhi(\bbx, \bbS_m; \ccalH_m)=\bbX_L \in \mathbb{R}^{m \times F_L}$. The use of multiple features does not affect our analysis.
\end{remark}
\subsection{Scalability of Graph Neural Networks on Wireless Networks}
\label{sec:scalabilityofwirelessnetworks}
To solve resource allocation on a small network $\bbS_m^1$, we train a GNN $\bbPhi(\bbx_m, \bbS_m; \mathcal{H}_m)$ that can be deployed on any network of size $m$. In practice, it can be useful for the trained model to maintain good performance on other networks $\bbS_m^2$. This property is called transferability, which requires that the following difference is bounded:
\begin{align}\label{eq:transferability}
    \Tilde{\Delta}(\bbS_{m}^1, \bbS_{m}^2, \ccalH_m):=\frac{1}{m}\|\bbPhi(\bbx_m, &\bbS_{m}^1; \ccalH_m)\nonumber\\
    &- \bbPhi(\bbx_m, \bbS_{m}^2; \ccalH_m)\|^2.
\end{align}

The parameters of a GNN trained for graphs of $m$ nodes can be expressed with dimensions that are independent of the size of the problem. Therefore, a GNN trained on $m$ nodes can run on graphs of any arbitrary size. In this case, the allocation has to remain robust when migrating from smaller training topologies to larger testing environments. This is characterized by the scalability property, measured as follows. 
\begin{align}\label{eq:scalability}
    \Delta(\bbS_m, \bbS_n, \ccalH_m):=\frac{1}{m}\|\sqcap_m[\bbPhi(\bbx_n, &\bbS_{n}; \ccalH_m)]\nonumber\\
    &- \bbPhi(\bbx_m, \bbS_{m}; \ccalH_m)\|^2,
\end{align}
where $\sqcap_m$ is a windowing operation that chooses a subset of nodes from $\bbS_n$, with $n > m$ such that it is statistically close to $\bbS_m$ (see Ass. \ref{ass:windowingoperationDGGs} and \ref{ass:windowingoperationRGGs}).

The value of $\Delta(\bbS_m, \bbS_n, \ccalH_m)$ indicates the difference in running the model on a large graph and then projecting its outputs to a smaller graph as compared with running the model on the small graph. We say a model presents scalability if $\Delta(\bbS_m, \bbS_n, \ccalH_m)$ is bounded by a small constant. Such property is useful in practical applications as training can be done in small networks to later deploy in larger layouts.

The central objective of this work is to prove that the scalability gap $\Delta$ remains bounded as the testing network size increases for sparse graphs and conflict graphs.
We evaluate this within the context of Random Geometric Graphs, which capture the spatial deployment and sparsity of wireless networks, employing a Deterministic Grid Graph model as a proxy.

%% file: 03_transferability_gnns.tex
We consider Deterministic Grids and Random Geometric Graphs for our analysis. The former has a predefined grid-like structure, while the latter samples node positions from a probability distribution and creates edges based on them. We denote a graph and its adjacency matrix with the same symbol. In particular, we use $\bbG_m$ for DGGs and $\bbR_m$ for RGGs of size $m$.

\begin{definition}[Deterministic Grid Graph]
    Consider a grid of $m$ points,\footnote{While the analysis can be extended to general settings, we assume $\sqrt{m}\in\mathbb{Z}$ for ease of presentation.} each with position $u=(i,j), i, j\in\{1, 2, \dots, \sqrt{m}\}$, forming the set $U$. A Deterministic Grid Graph is a graph with node set $U$ and edges $(u, u')=((i,j), (i', j'))$ connecting nodes with either $i'=i\pm1$ or $j'=j\pm1$.
\end{definition}

A Random Geometric Graph is defined over a metric space—typically the Euclidean plane—where node coordinates are drawn from a spatial probability distribution \cite{penrose2003random}. 

\begin{definition}[Random Geometric Graph]
    Consider a set of points $U\in\mathbb{R}^2$ sampled uniformly over a metric space. A Random Geometric Graph is a graph with node set $U$ and edges $(u, u')$ if and only if $\|u-u'\|_2\leq r$. Each edge is weighted with $w((u, u')) =\|u-u'\|_2$.
\end{definition}

Figure \ref{fig:rgg-dgg-complete} illustrates DGGs and RGGs at various scales. Observe that an RGG can be obtained by adding noise to the node positions of a DGG and later drawing the edges according to the radius $r$. Given the positions of the nodes, the graphs can be very different. For our theoretical results, we assume closeness between their topologies. 

\begin{assumption}[Closeness of graphs]\label{ass:closeness}
    Graphs $\bbG_m, \bbR_m$ are structurally close if the difference between their adjacency matrices in the spectral norm is bounded: 
    \begin{align}
        ||\bbG_m-\bbR_m||\leq \varepsilon.
    \end{align}
\end{assumption}
Note that the noise added to the positions of a DGG to construct an RGG controls the closeness between the two graphs. We show the scalability of GNNs over RGGs by combining the scalability of DGGs with their structural closeness to RGGs.

\begin{figure}[t] 
\centering
\includegraphics[width=0.95\linewidth]{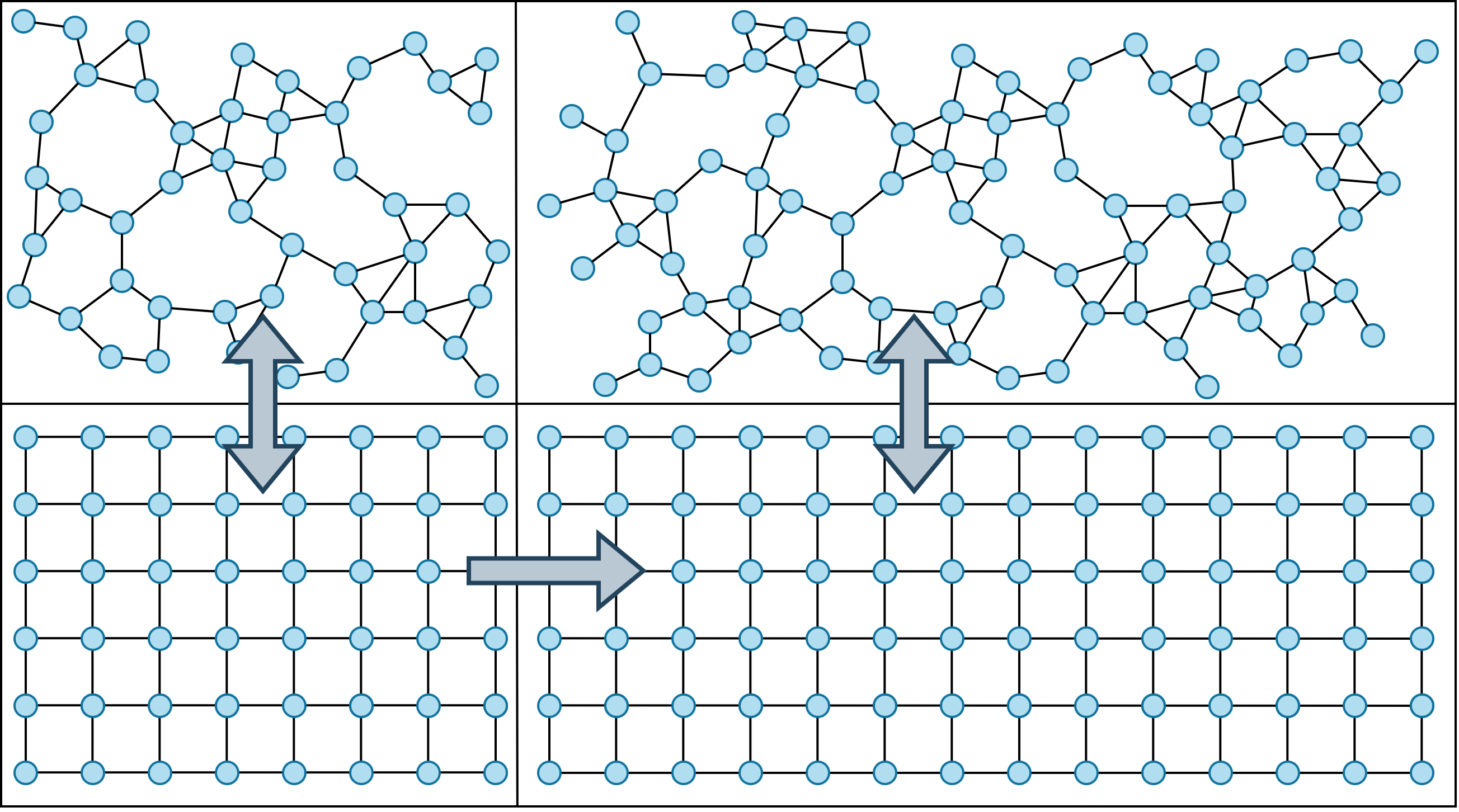}
\caption{Illustration of Random Geometric Graphs (top) and Deterministic Grid Graphs (bottom) of different sizes. The arrows outline the proposed analysis to establish Graph Neural Network scalability.}
\label{fig:rgg-dgg-complete}
\end{figure}

\subsection{Scalability of Graph Neural Networks on DGGs}
\label{sec:transfdggdgg}
Let $\bbG_{m}$ be a DGG with $m$ nodes. Similarly, let $\bbG_{n}$ denote a larger DGG containing $n$ nodes. Because their adjacency matrices possess a block-circulant structure, graph convolutions on a DGG can be reinterpreted as standard 2D spatial convolutions. Executing a GNN on a DGG is mathematically analogous to applying a 2D Convolutional Neural Network (CNN) to an image, which have been shown to exhibit scalability properties \cite{owerko2023transferabilityconvolutionalneuralnetworks, madala23cnnscurseofdimensionality}. 

Following our definition of scalability (cf. \eqref{eq:scalability}), we work on the larger dimension $n$, with a signal $\bbx_n\in\mathbb{R}^n$. We consider a selection window $\sqcap_m$ that extracts $m$ components to construct features $\bbx_m=\sqcap_m\bbx_n \in \mathbb{R}^n$ and the graph $\bbG_m = \sqcap_m \bbG_n \sqcap_m\in\mathbb{R}^{n\times n}$. 

\begin{assumption}\label{ass:windowingoperationDGGs}
    The selection window $\sqcap_m$ is a truncation operation such that $\bbG_m = \sqcap_m \bbG_n \sqcap_m$ is a DGG. 
\end{assumption}
Assumption \ref{ass:windowingoperationDGGs} ensures the comparison between graphs with similar structures. To satisfy this, the windowing operation selects $m$ nodes from $\bbG_n$ that are spatially adjacent.

Our theoretical results exploit properties of smoothness in the context of signals, of the GNN architecture and the structure of the graphs. We formalize our assumptions as follows, beginning with stochastic properties of the data.

\begin{assumption}(Stationarity)\label{ass:stationarity}
    The graph signal $\bbx_n$ is bounded and stationary with respect to the underlying graph shift operator $\bbG_n$. Specifically, we assume its covariance matrix $\bbC_\bbx = {\bbx_n\bbx_n^\top}$ commutes with the GSO, i.e., $\bbG_n\bbC_\bbx = \bbC_\bbx \bbG_n$ \cite{marques2017stationary, perraudin2017stationary}.
\end{assumption}

Assumption \ref{ass:stationarity} ensures that the statistics of the signals remain consistent across scales. Because the global signal is stationary, any windowed version preserves these spatial statistics.

We impose smoothness on the architecture by restricting the pointwise nonlinearity. In particular, we consider normalized Lipschitz nonlinearities:

\begin{assumption}[Normalized Lipschitz nonlinearity]\label{ass:lipschitznonlinearity}
    A pointwise nonlinearity $\gamma$ is normalized Lipschitz continuous if $\gamma(0)=0$ and $|\gamma(a)-\gamma(b)|\leq |a-b|$ for all $a, b \in \mathbb{R}$.
\end{assumption}
Common activation functions such as ReLU and sigmoid satisfy Assumption \ref{ass:lipschitznonlinearity}. 

Let $\bbPhi$ denote an $L$-layer GNN with a graph filter $\bbh$ and a pointwise nonlinearity $\gamma$. The architecture is trained on graphs of size $m$, with parameters $\ccalH_m$. We present GNN scalability over DGGs in Theorem \ref{th:dggonndgg}.
\begin{theorem}\label{th:dggonndgg}
    Let $\bbPhi(\bbx_n, \bbG_n; \ccalH_m)$ and $\bbPhi(\bbx_m, \bbG_m; \ccalH_m)$ denote the parameterized policies evaluated on DGGs with $n$ and $m$ nodes and normalized adjacency matrices. With Assumptions \ref{ass:windowingoperationDGGs}--\ref{ass:lipschitznonlinearity}, the scalability (cf. \eqref{eq:scalability}) is bounded as 
    \begin{align}\label{eq:boundfordggs}
        \Delta(\bbG_m, \bbG_n, \ccalH_m) \leq \frac{H_{L,K}^2}{m} \left( 2\sqrt{m}LK + L^2K^2 \right) \bbx_\text{max}^2,
    \end{align}
    where $H_{L,K} = \max_{l \in \{1, \dots, L\}} \sum_{k=0}^{K} |h_{l, k}|$ and $\bbx_\text{max}=\|\bbx_m\|_\infty$. 
\end{theorem}
\begin{proof}
    See Supplementary Material \ref{sec:appendixproofdggtodgg}.
\end{proof}

Theorem \ref{th:dggonndgg} establishes that a GNN trained on small DGGs can be deployed on larger DGGs with bounded performance degradation. The bound captures a structural limitation of finite graphs: nodes situated near the perimeter of $\bbG_m$ lack complete $LK$-hop neighborhoods, resulting in truncated filtering operations. Transferring a policy trained on a finite domain inherently incurs this localized approximation error at the borders.

Equation \eqref{eq:boundfordggs} guarantees asymptotic scalability. In the limit, the proportion of boundary nodes vanishes, driving the boundary effect constant—and the generalization gap—to zero. This confirms that while small training graphs are sensitive to boundary effects, moderately sized networks possess enough interior nodes to dominate the error and enable scalability. 

\begin{remark}
    Theorem \ref{th:dggonndgg} shows a bound for scalability of GNNs on DGGs. Training on larger DGGs provably can narrow the scalability gap and enable the transferability to larger scale DGGs. This can be seen as an extended version of previous results on transferability studies of standard CNNs \cite{owerko2023transferabilityconvolutionalneuralnetworks}. In our work, we adapt the analysis in the discrete domain and consider unsupervised learning tasks, leveraging the equivalence between graph and regular convolutional filters. Via the windowing operation $\sqcap_m$, we compare the outputs of the GNN directly applied to graphs of different scales.
\end{remark}

\subsection{Transferability of Graph Neural Networks between RGG and DGG}
\label{sec:transfrggdgg}
Let $\bbG_{m}$ be a DGG with $m$ nodes. We define an RGG, $\bbR_m$, generated by applying two-dimensional Gaussian noise to the node coordinates of $\bbG_m$. Edges in $\bbR_m$ are formed between any pair of nodes whose perturbed positions fall within a fixed radius $r$. For $\bbG_m$ to serve as a reliable analytical proxy for $\bbR_m$, their topologies must be sufficiently similar, as presented in Assumption \ref{ass:closeness}. 

To guarantee architectural stability, we combine Assumption \ref{ass:lipschitznonlinearity} with smoothness of the graph filters. Let us define the frequency response of a filter $\bbh$ of order $K$, denoted $\hat{h}(\lambda)$ \cite{segarra2017optimalgraphfilter, sandryhalla2014dspongraphs}:
\begin{align}\label{freqresponsefilter}
    \hat{h}(\lambda) = \sum_{k=0}^{K}h_k\lambda^k.
\end{align}

Equation \eqref{freqresponsefilter} shows that the graph frequency response depends exclusively on the filter coefficients, which characterize the filtering operation in the spectral domain. Our next assumption requires smoothness in the frequency response of the filter.

\begin{assumption}[Integral Lipschitz continuous filter]\label{ass:lipschitzfilter}
    A graph filter $\bbh$ is integral Lipschitz continuous if its frequency response $\hat{h}$ satisfies
    \begin{equation}\label{eqn:filter_function_integral}
        |\hat{h}(a)-\hat{h}(b)|\leq \frac{C |a-b| }{(a+b)/2} \quad \text{for all } a,b \in (0,\infty),
    \end{equation}
    where $C>0$ is the Lipschitz constant.
\end{assumption}
The integral Lipschitz property ensures that the filter remains stable under structural deformations and domain perturbations of the underlying graph \cite{gama2020stability}. 

Under Assumptions \ref{ass:closeness}, \ref{ass:stationarity}--\ref{ass:lipschitzfilter}, a GNN $\bbPhi$ trained on RGGs can be executed on DGGs with minimal performance loss. Theorem \ref{the:rgg-gg-gnntransf} bounds the performance gap when transferring a GNN between an RGG and a DGG of the same scale, with input $\bbx_m$.

\begin{theorem}\label{the:rgg-gg-gnntransf}
    Let $\bbR_m$ and $\bbG_m$ be an RGG and a DGG with $m$ nodes, with normalized adjacency matrices. A GNN $\bbPhi(\cdot, \cdot; \ccalH_m)$ is trained with parameters $\ccalH_m$, under Assumptions \ref{ass:closeness}, \ref{ass:stationarity}--\ref{ass:lipschitzfilter}. The degradation in performance when transferring the model from RGGs to DGGs (cf. \eqref{eq:transferability}) is bounded,
    \begin{equation}\label{eq:boundrggtodgg}
        \Tilde{\Delta}(\bbR_{m}, \bbG_{m}, \ccalH_m) \leq 4L^2\varepsilon C\bbx_\text{max}^2.
    \end{equation}
\end{theorem}
\begin{proof}
    See Supplementary Material \ref{sec:appendixproofrggtodgg}.
\end{proof}

Theorem \ref{the:rgg-gg-gnntransf} shows that smooth architectural operations enable GNN transferability between structurally close graphs. The error bound in Equation \eqref{eq:boundrggtodgg} remains tight and meaningful under standard design paradigms. First, the input signal $\bbx_m$ is bounded. Second, GNN architectures typically employ few layers to prevent oversmoothing and the unintended mixing of information from distant neighborhoods. Third, the filter's integral Lipschitz constant $C$ governs a trade-off: a smaller $C$ tightens the transferability bound but restricts the expressivity of the frequency response. Finally, the topological distance $\varepsilon$ can be reduced by evaluating increasingly similar graph structures.

The results from Theorem \ref{the:rgg-gg-gnntransf} bound the transferability gap between RGGs and DGGs. Whether we train on RGGs and transfer to DGGs or vice versa, the performance bound remains the same. This, when combined with our results for scalability of GNNs over DGGs (Section \ref{sec:transfdggdgg}), is what enables the extension to scalability of GNNs over RGGs, as presented next.

\begin{remark}
    In our formulation, RGGs are weighted by the Euclidean distance between nodes, whereas DGGs require uniform edge weights across all grid connections. Uniform weights are essential to preserve the block-circulant matrix structure and spatial shift-invariance. Because our analysis evaluates the graphs via the spectral norm of their normalized adjacency matrices, the discrepancy between the uniform DGG and the weighted RGG is captured by the closeness parameter $\varepsilon$ (see Ass. \ref{ass:closeness}). To see how the perturbation changes the eigenspaces we leverage the Davis-Kahan theorem \cite{daviskahan}.
\end{remark}

\subsection{Scalability of Graph Neural Networks over RGG over scales}
\label{sec:transfrgg}
Section \ref{sec:transfdggdgg} demonstrated that a GNN trained on DGGs of size $m$ scales to larger DGGs of size $n$ with bounded performance loss. Moreover, Section \ref{sec:transfrggdgg} established that GNNs trained on RGGs of size $m$ can be transferred to DGGs of the same scale, provided their topologies are sufficiently close. The latter result also guarantees that the performance on a DGG of size $n$ tightly approximates the performance on an RGG of the same size $n$. Together, these bounds close the theoretical loop illustrated in Figure \ref{fig:rgg-dgg-complete}, formally proving that GNNs can be transferred across RGGs of expanding scales with minimal performance degradation. 

The windowing operation for DGGs in Section \ref{sec:transfdggdgg} crops a sub-grid from a larger grid, as presented in Assumption \ref{ass:windowingoperationDGGs}. The windowing operator $\sqcap_m$ applied to a large RGG $\bbR_n$ must isolate a localized subgraph that is statistically close to a small RGG $\bbR_m$, preserving topological properties such as the degree distribution and local edge connectivity. We describe this requirement in the following assumption.

\begin{assumption}\label{ass:windowingoperationRGGs}
      Consider a graph $\bbR_n$ that satisfies Assumption \ref{ass:closeness}, close to a DGG $\bbG_n$. The selection window $\sqcap_m$ is such that $\bbR_m = \sqcap_m \bbR_n \sqcap_m$ satisfies Assumption \ref{ass:closeness} close to $\bbG_m$, where $\bbG_m = \sqcap_m \bbG_n \sqcap_m$ is a truncated DGG. 
\end{assumption}
This ensures the resulting smaller RGG is also close to a DGG. If we generate RGGs as spatial perturbations of DGGs, Assumption \ref{ass:windowingoperationRGGs} is satisfied if we select an adjacent $m$-node section of the graph.

We formalize the end-to-end analysis in Theorem \ref{the:rggtonrgg}, our primary result establishing the transferability of GNNs across scaling RGGs.

\begin{theorem}\label{the:rggtonrgg}
    Under Assumptions \ref{ass:closeness}--\ref{ass:windowingoperationRGGs}, a GNN $\bbPhi$ is trained on an RGG of size $m$, with parameters $\ccalH_m$. The model can be executed on larger RGGs of size $m$ with bounded performance loss, showing scalability:
    \begin{align}
         \Delta(\bbR_m, \bbR_n, \ccalH_m) \leq& 24L^2\varepsilon C \bbx_{\text{max}}^2\nonumber\\
    &+\frac{3H_{K,L}^2}{m}(2\sqrt{m}LK+L^2K^2)\bbx_{\text{\text{max}}}^2
    \end{align}
\end{theorem}
\begin{proof}
    See Supplementary Material \ref{sec:appendixproofrggtorgg}.
\end{proof}

Theorem \ref{the:rggtonrgg} demonstrates that a GNN trained on a small RGG can be deployed on a larger RGG with a bounded performance gap. This bound is governed by the structural approximation error ($\varepsilon$), and the energy of the input signal, which we assume remains bounded as the network dimension increases. The number of layers $L$ and the order of the filter $K$ are small in practice.

These results extend the existing literature by establishing that GNNs can be successfully transferred across scales for sparse graphs. In particular, by leveraging Random Geometric Graphs as a model for wireless networks, we demonstrate that GNN-based resource allocation policies can be trained on small-scale networks and subsequently deployed on significantly larger networks with bounded performance degradation. 

\begin{remark}
    The theory developed in this section is tied to our assumptions about smoothness. The result is valid when the graph filter is integral Lipschitz and the RGGs considered are close to DGGs. The assumptions on the stationarity of the signals and the Lipschitz property of the pointwise nonlinearity can be easily met and are common. In our numerical experiments (Section \ref{sec:numericalexperiments}) we study how relevant our hypothesis are for practical results. 
\end{remark}

%% file: 04_transferability_gnns_conflicts.tex
We now extend our scalability analysis for GNNs on Random Geometric Graphs to their corresponding conflict graphs. We denote the conflict graph of an RGG $\bbR_m$ with $c$ edges as $\hat{\bbR}_c$.

\begin{definition} (Conflict Graph)\label{def:conflictgraphreal}
    Given a graph $\bbR_m = (\ccalV, \ccalE)$ with $|\ccalV| = m$ nodes and $|\ccalE| = c$ edges $e = (i,j) \in \ccalE$, its conflict graph $\hat{\bbR}_c$ is an unweighted graph with node set $\ccalE$ whose adjacency matrix elements satisfy:
    \begin{align}
        [\hat{\bbR}_c]_{e, e'} = 1 \iff e \neq e' \text{ and } \{i, j\} \cap \{i', j'\} \neq \emptyset.
    \end{align}
\end{definition}

RGGs and their conflict graphs under varying spatial perturbation parameter $\sigma$ are illustrated in Figure \ref{fig:rggsandconflicts}. While the underlying topology $\bbR_m$ cannot be uniquely reconstructed from its conflict graph $\hat{\bbR}_c$, edge-supported signals on $\bbR_m$ map to node-level representations on $\hat{\bbR}_c$. Formally, interpreting the conflict graph as a line graph allows the spectral properties (eigendecomposition) of $\bbR_m$ and $\hat{\bbR}_c$ to be related via the incidence matrix and their respective vertex and edge Laplacian representations.

\begin{figure*}
    \centering
    \includegraphics[width=0.95\linewidth]{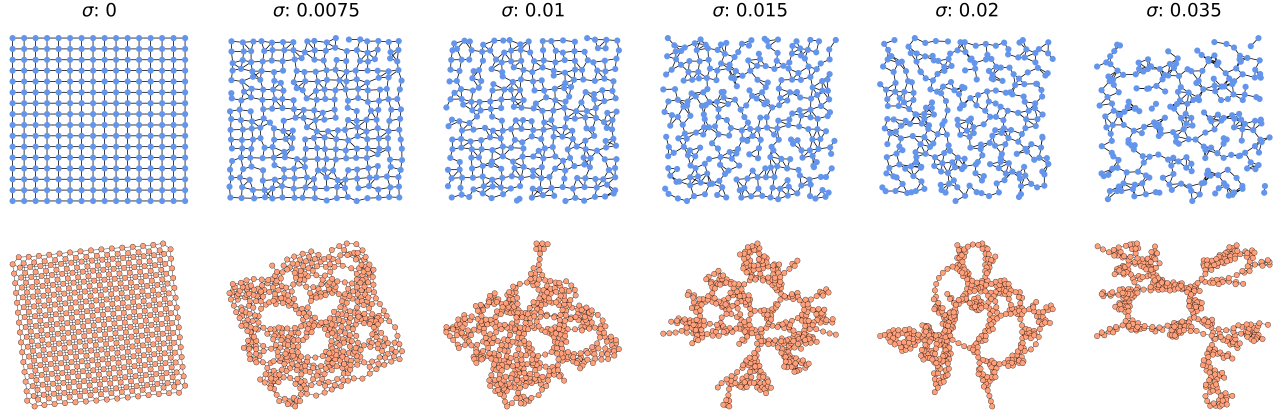}
    \caption{Illustration of Random Geometric Graphs (top) and their corresponding conflict graphs (bottom). From left to right, the positional noise added to Deterministic Grid Graphs, $\sigma$, increases.}
    \label{fig:rggsandconflicts}
\end{figure*}

\begin{remark}
    Conflict graphs are often used to model edge-supported tasks, such as link scheduling in wireless communications and resource collision constraints in integer linear programming. In these applications, the conflict graphs are unweighted, as edges represent discrete resources or binary decision variables. Our analysis focuses on this unweighted regime.
\end{remark}

\subsection{Scalability of Graph Neural Networks on conflict graphs of DGGs}
\label{sec:dggtondgg-conflict}
Let $\hat{\bbG}_c$ and $\hat{\bbG}_d$ be the conflict graphs of DGGs with $c$ and $d$ nodes, respectively. As illustrated in Figure \ref{fig:rggsandconflicts} for $\sigma=0$, the conflict graph of a DGG maintains a highly regular structure. However, the degree of the interior nodes increases from 4 to 6, and the boundary width expands. 

Despite the increased density, the adjacency matrices of these conflict graphs retain a block-circulant structure. This allows us to analyze their scalability similarly to that of CNNs \cite{owerko2023transferabilityconvolutionalneuralnetworks}, with adjustments for the conflict graph's geometry. Specifically, the boundary perimeter of the conflict graph no longer scales with the constant $2\sqrt{c}$ found in standard square grids, but rather scales proportionally to $4\sqrt{2c}$. Here, $4\sqrt{2}$ considers the increase in boundary thickness to 2, which accounts for the two concentric layers along the perimeter of $\hat{\bbG}_c$ with degrees ranging from 3 to 5, and the truncated $LK$-hop convolutional neighborhoods.

The results for the scalability of GNNs on these conflict graphs follow the logic established in Theorem \ref{th:dggonndgg}. We operate in the larger dimension $d$ and utilize a windowing operation $\sqcap_c$ to isolate components of the larger graph, ensuring the topology matches the smaller graph (see Ass. \ref{ass:windowingoperationDGGs}). With an input signal $\bbx_c = \sqcap_c \bbx_d \in \mathbb{R}^d$ satisfying Assumption \ref{ass:stationarity}, we compare the GNN's output on $\hat{\bbG}_d$ against $\hat{\bbG}_c = \sqcap_c \hat{\bbG}_d \sqcap_c \in \mathbb{R}^d$. We impose architectural stability on the $L$-layer GNN, $\bbPhi$, trained on scale $c$ via Assumption \ref{ass:lipschitznonlinearity}.

\begin{theorem}\label{th:dggonndgg-conflicts}
    Let $\bbPhi(\bbx_c, \hat{\bbG}_c; \ccalH_c)$ and $\bbPhi(\bbx_d, \hat{\bbG}_d; \ccalH_c)$ denote the parameterized policies evaluated on the conflict DGGs with $c$ and $d$ nodes and normalized adjacency matrices. Under Assumptions \ref{ass:windowingoperationDGGs}--\ref{ass:lipschitznonlinearity}, the scalability is bounded as
    \begin{align}\label{eq:boundfordggs-conflict}
        \Delta(\hat{\bbG}_c, \hat{\bbG}_d, \ccalH_c) \leq \frac{H_{L,K}^2}{c} \left( 4\sqrt{2c}LK + L^2K^2 \right) \bbx_\text{max}^2.
    \end{align} 
\end{theorem}
\begin{proof}
    See Supplementary Material \ref{sec:appendixproofdggtodgg-conflict}.
\end{proof}

The bound presented in Theorem \ref{th:dggonndgg-conflicts} for GNN scalability on conflict graphs is similar to that on DGGs. It guarantees that as the network size $c$ grows, the performance gap between the small-scale training graph and the large-scale testing graph asymptotically vanishes. This holds true despite the conflict graph being denser than the original DGG. Moreover, as $c$ grows, the border effects become insignificant.

\subsection{Transferability of Graph Neural Networks between conflict graphs of RGGs and DGGs}
\label{sec:rggtodgg-conflict}
Consider an RGG $\bbR_m$ and a reference DGG $\bbG_m$, where $\bbR_m$ is generated via a perturbation of $\bbG_m$. Because random node displacements alter the total edge count, the resulting conflict graphs may differ in dimension. To evaluate transferability at equal scale, we embed their normalized adjacency matrices into a common dimension $c = \max(c_{\mathbf{R}}, c_{\mathbf{G}})$ via zero-padding, yielding $\hat{\bbR}_c, \hat{\bbG}_c \in \mathbb{R}^{c \times c}$. Zero-padding introduces isolated nodes with zero eigenvalues, leaving the spectral norm difference intact.

The adjacency matrix of a conflict graph satisfies $\hat{\mathbf{A}} = \mathbf{B}^\top \mathbf{B} - 2\mathbf{I}$, where $\mathbf{B}$ denotes the unoriented incidence matrix of the original graph. Because the non-zero spectra of the signless edge Laplacian $\mathbf{B}^\top \mathbf{B}$ and signless vertex Laplacian $\mathbf{B}\mathbf{B}^\top$ coincide, closeness between the base graphs ($\|\bbG_m - \bbR_m\| \le \epsilon$) directly controls the spectral distance between their conflict graphs, establishing $\|\hat{\bbG}_c - \hat{\bbR}_c\| \le \epsilon$ (Assumption \ref{ass:closeness}).

Considering the assumption of structural closeness, Theorem \ref{the:rgg-gg-gnntransf-conflict} bounds the performance degradation incurred when transferring an $L$-layer GNN $\bbPhi(\cdot, \cdot; \ccalH_c)$ between the conflict graphs of an RGG and a DGG of the same scale.

\begin{theorem}\label{the:rgg-gg-gnntransf-conflict}
    Let $\hat{\bbR}_c$ and $\hat{\bbG}_c$ be conflict graphs of an RGG and a DGG, respectively, with normalized adjacency matrices zero-padded to dimension $c$. Under Assumptions \ref{ass:closeness}, \ref{ass:stationarity}--\ref{ass:lipschitzfilter}, the performance gap when deploying a GNN $\bbPhi(\cdot, \cdot; \ccalH_c)$ trained on $\hat{\bbR}_c$ onto $\hat{\bbG}_c$ is bounded as
    \begin{equation}\label{eq:boundrggtodgg-conflict}
        \tilde{\Delta}(\hat{\bbR}_{c}, \hat{\bbG}_{c}, \ccalH_c) \leq 4L^2\varepsilon C\bbx_\text{max}^2.
    \end{equation}
\end{theorem}
\begin{proof}
    See Supplementary Material \ref{sec:appendixproofrggtodgg-conflict}.
\end{proof}

As was the case with Theorem \ref{the:rgg-gg-gnntransf}, Theorem \ref{the:rgg-gg-gnntransf-conflict} shows the transferability of GNNs between conflict RGGs and conflict DGGs training on either structure. 

\begin{remark}
    While the underlying RGGs may incorporate continuous edge weights, their associated conflict graphs are unweighted. The incidence matrix $\mathbf{B}$ is constructed directly from the graph's connectivity structure. This guarantees that the signless Laplacian identity $\hat{\mathbf{A}} = \mathbf{B}^\top \mathbf{B} - 2\mathbf{I}$ holds, ensuring that topological closeness of the original graphs results in spectral closeness of their conflict graphs.
\end{remark}

\subsection{Scalability of Graph Neural Networks on conflict graphs of RGGs}
\label{sec:rggtonrgg-conflict}
Section \ref{sec:dggtondgg-conflict} established that GNN policies scale across conflict graphs of DGGs. Later, Section \ref{sec:rggtodgg-conflict} bounded the transferability gap between conflict graphs of RGGs and DGGs at any given scale. By combining these intermediate bounds, we close the theoretical loop, proving that GNN policies trained on small conflict RGGs generalize to large-scale conflict RGGs with bounded performance loss.

Theorem \ref{the:rggtonrgg-conflict} presents our main theoretical result for conflict graphs, establishing the transferability of GNNs across scales for conflict graphs of RGGs.

\begin{theorem}\label{the:rggtonrgg-conflict}
    Under Assumptions \ref{ass:closeness}--\ref{ass:windowingoperationRGGs}, let a GNN $\bbPhi(\cdot, \cdot; \ccalH_c)$ be trained on conflict graphs of RGGs with $c$ nodes, yielding parameters $\ccalH_c$. When deployed without retraining onto a larger conflict graph $\hat{\bbR}_d$ of scale $d > c$, the scalability error is bounded as
    \begin{align}\label{eq:boundrggtonrgg-conflict}
        \Delta(\hat{\bbR}_c, \hat{\bbR}_d, \ccalH_c) \leq& 24L^2\varepsilon C \bbx_{\text{max}}^2 \nonumber\\&+ \frac{3H_{L,K}^2}{c}\left(4\sqrt{2c}LK + L^2K^2\right)\bbx_{\text{max}}^2.
    \end{align}
\end{theorem}
\begin{proof}
    See Supplementary Material \ref{sec:appendixproofrggtorgg-conflict}.
\end{proof}

The bound for scalability of GNNs over conflict RGGs in Equation \eqref{eq:boundrggtonrgg-conflict} decreases with $1/\sqrt{c}$. As the size of the conflict graphs used for training increases, the transferability to larger graphs is further improved.

%% file: 05_numerical_experiments.tex
We evaluate the scalability of GNNs on two representative network management tasks: Power Allocation (PA) (see Ex. \ref{ex:powerallocation}) and Wireless Link Scheduling (WLS) (see Ex. \ref{ex:wls}). The former is a task supported on the nodes of RGGs, while the latter is supported on the nodes of the conflict graph of RGGs. Details regarding the primal-dual learning framework and algorithmic formulations are provided in \cite{eisen2020optimal, camargo2026longhorizonwirelesslinkscheduling}.

For both applications, we employ an $L=3$ layer GNN architecture built upon TAGConv graph convolutional layers \cite{tagconv} using graph filters of order $K=1$. The multi-feature filters are set to $F=8$ for PA and $F=256$ for WLS. Model parameters are optimized via Adam \cite{adam}, with learning rates detailed in Supplementary Material \ref{sec:expdetails} for different perturbation levels. 

Datasets are generated by constructing DGGs at the target scale and applying spatial perturbations to obtain RGGs. For WLS, scaling the number of links in the communication graph controls the node count of the resulting conflict graphs. Our theoretical guarantees depend on the perturbation parameter $\varepsilon$. In practice, we denote this with $\sigma$, which quantifies the variance of Gaussian noise applied to the positions of DGGs to obtain target RGGs. For PA, the graphs weighted with channel interference over different values of $\sigma$ are shown in Figure \ref{fig:parggsillustration}. For WLS, the reference values of $\sigma$ are those in Figure \ref{fig:rggsandconflicts}. Across all scales and perturbation levels, the average degree of the nodes is approximately constant, $d_{avg}\simeq6$. Unless otherwise stated, the models are trained on scale $m=500$, with $\sigma=0.004$ for PA and $\sigma=0.01$ for WLS. 

We benchmark our learned policies against heuristic baselines. For PA, we evaluate a discretized variation of WMMSE \cite{wmmse}. Given the continuous power vector from WMMSE, we select the $K=\lfloor P_{\max} /p_0\rfloor$  transmitters with the highest allocated values and map them to the discrete transmission power $p_0=5$, guaranteeing the total power is aligned with the budget $P_{\max}$ (set to the number of devices in the network). For WLS, we compare against FPLinQ \cite{shen2017fplinq}, which solves an instantaneous formulation of the scheduling problem. The minimum transmission requirement is set to $\bbdelta=0.1$, with $T=200$.

\subsection{Scalability of GNNs}
We first evaluate the transferability of the proposed GNN architecture trained on small-scale networks with $m=500$ nodes for PA. Figure \ref{fig:scalabilityPA} (top) illustrates the empirical distribution of per-user sum rates over 50 unseen test networks of the same scale ($m=500$). The trained GNN allocates higher rates across a larger fraction of links compared to the discretized WMMSE benchmark.

To assess scalability, the model trained on $m=500$ is evaluated on larger networks with $m=1500$ nodes (Figure \ref{fig:scalabilityPA}, bottom). The GNN retains its performance advantage over WMMSE across the scaled topology. While WMMSE requires re-running an iterative optimization algorithm for every new network realization, the GNN runs a single forward pass during inference.

\begin{figure}
    \centering
    \includegraphics[width=0.95\linewidth]{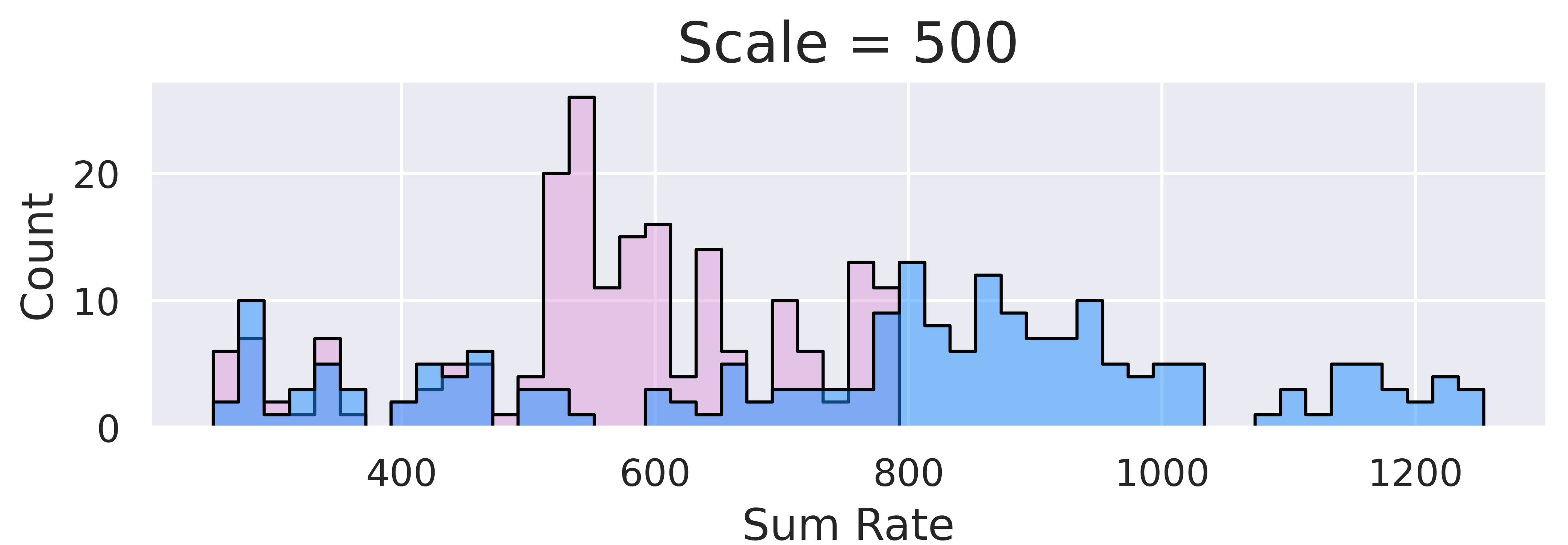}
    \includegraphics[width=0.95\linewidth]{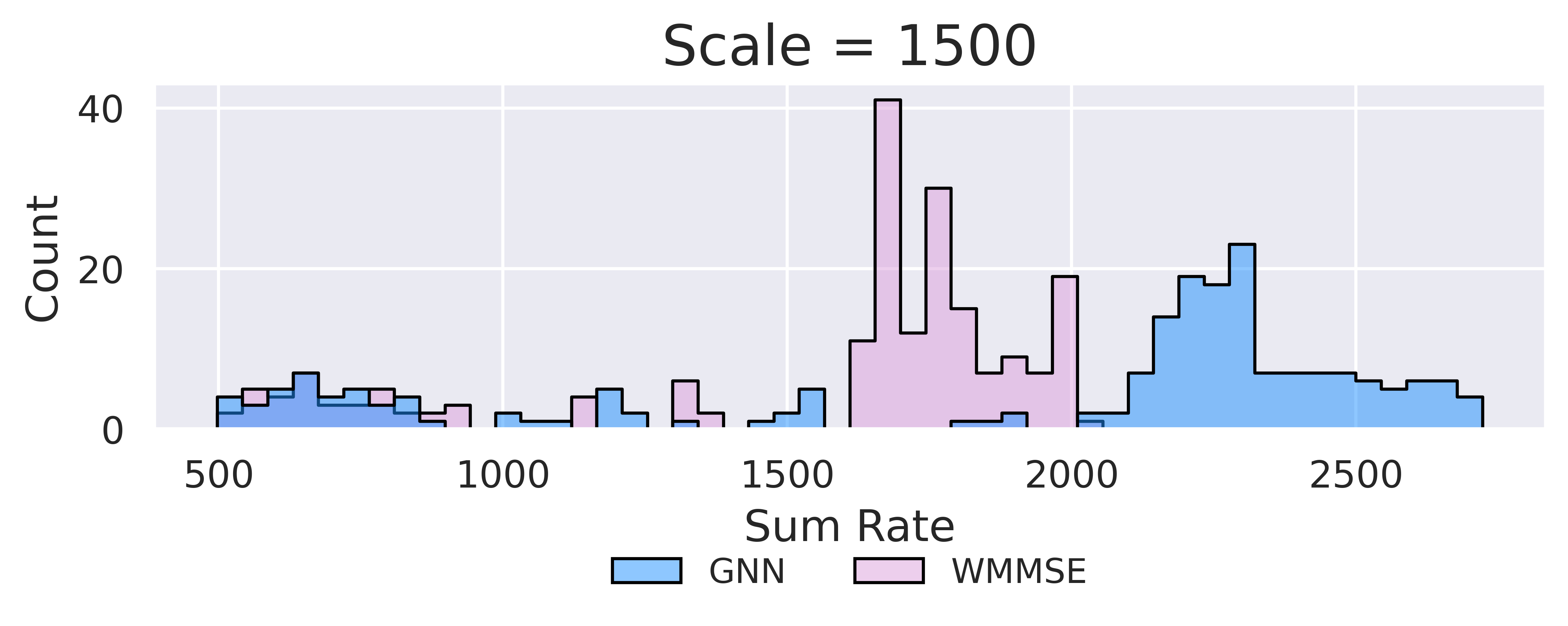}
    \caption{Sum rate histograms for a GNN trained to solve PA at scale $500$, evaluated on scale $m=1500$.}
    \label{fig:scalabilityPA}
\end{figure}


Figure \ref{fig:largermodelsscalebetterPA} compares the cross-scale generalization of three GNN models independently trained at scales $m\in\{500,1000,1500\}$. Models trained on larger graphs exhibit better scalability to higher-dimensional networks. To quantify feasibility, we define the normalized power constraint violation as $(P_{\text{total}}-P_{\max})/m$. As depicted in Figure \ref{fig:largermodelsscalebetterPA}, the model trained on $m=1500$ maintains constraint violations closest to zero when scaled to $m=2500$. However, this comes with a minor trade-off, as it does not perform as well when deployed onto smaller network scales ($m=500$).

\begin{figure}
    \centering
    \includegraphics[width=0.95\linewidth]{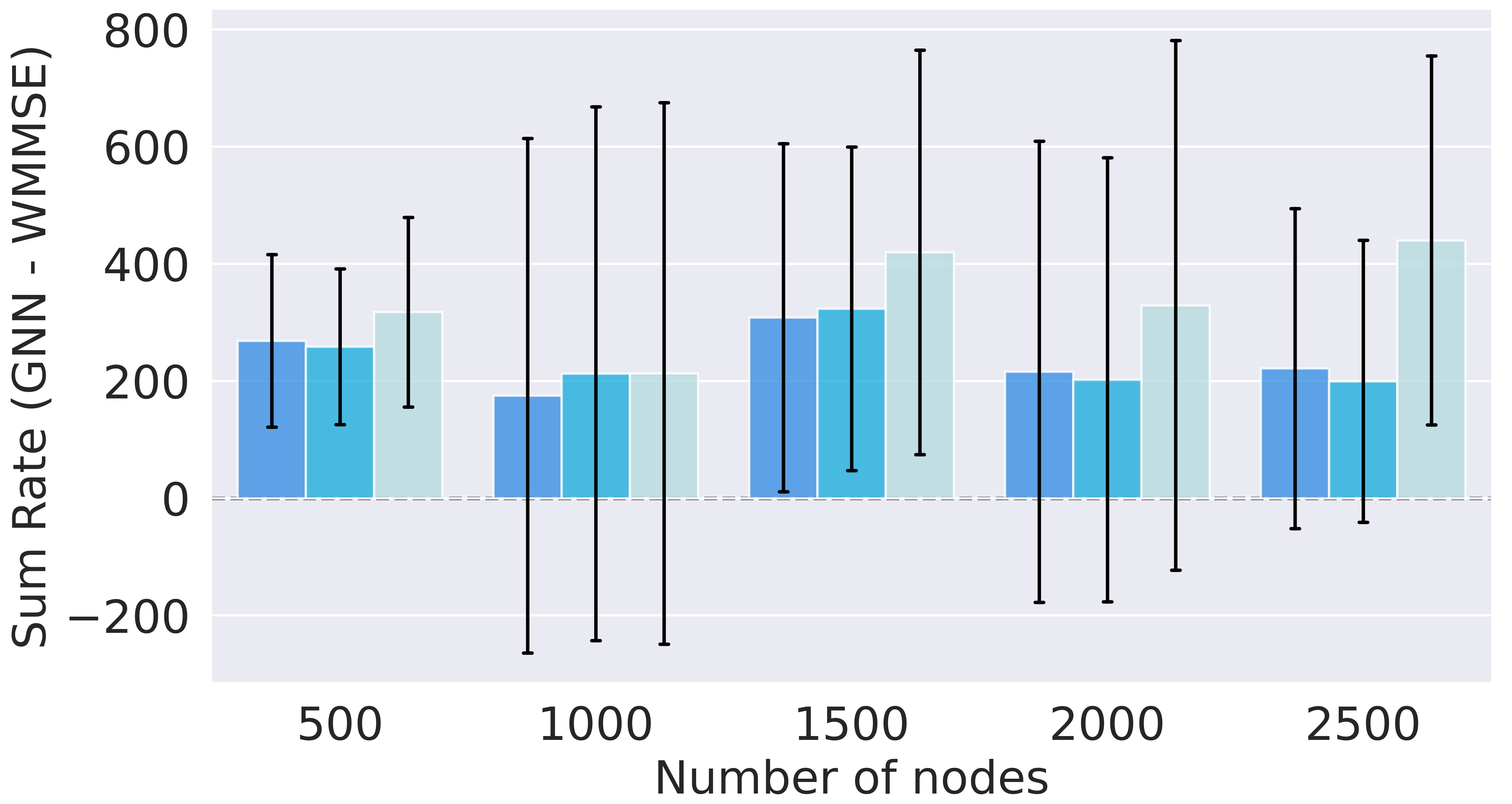}
    \includegraphics[width=0.95\linewidth]{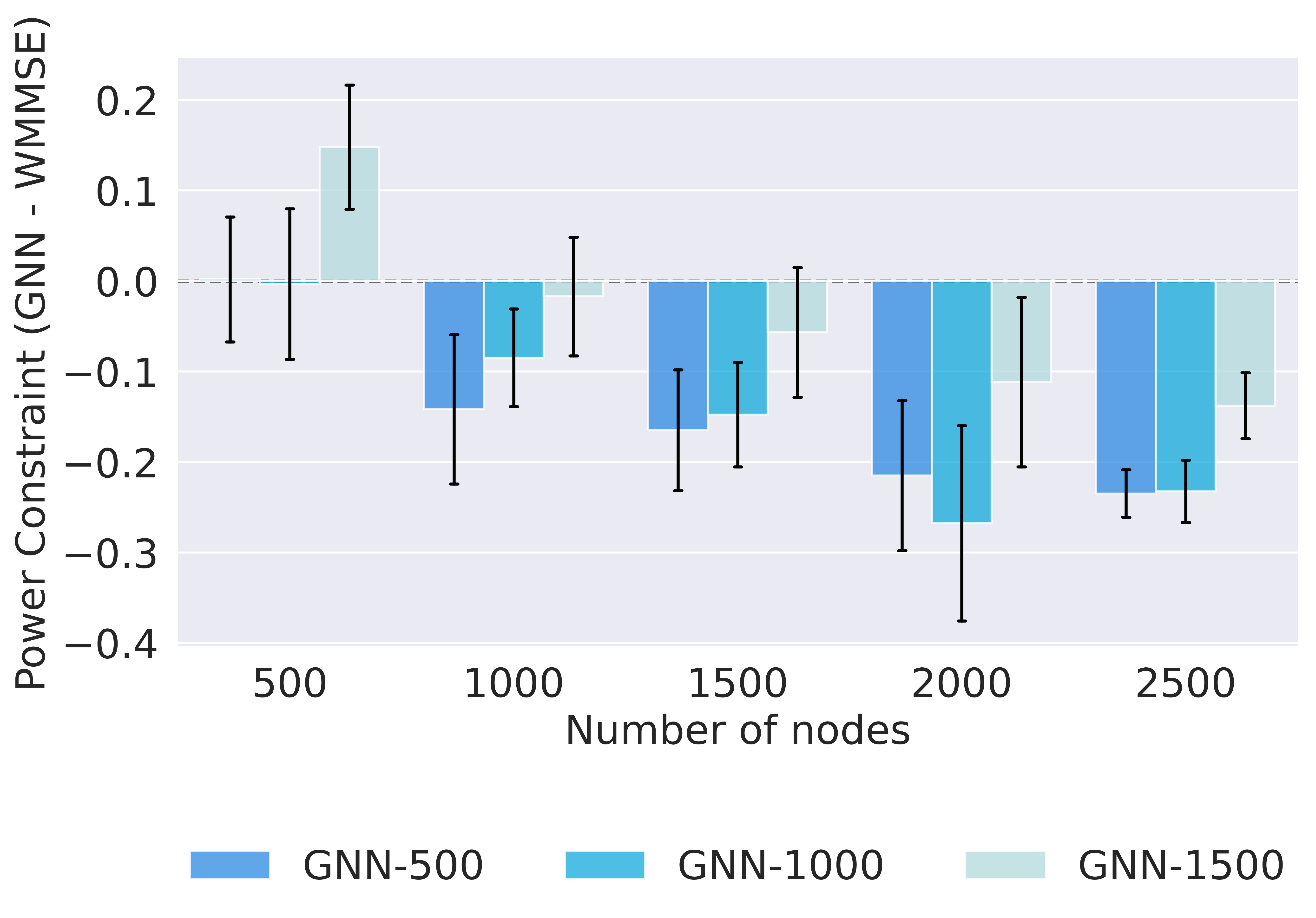}
    \caption{Sum rate and power constraint for models trained for PA on varying scales.}
    \label{fig:largermodelsscalebetterPA}
\end{figure}

For the WLS task, we train a GNN on conflict graphs with $c=500$ links and evaluate its performance against the FPLinQ baseline across growing networks (Figure \ref{fig:scalabilitysawl}). As network scale increases, the GNN policy generalizes, maintaining near-zero constraint violations and tightly matching FPLinQ’s sum-rate performance despite higher complexity. We emphasize that while FPLinQ optimizes instantaneous link rates, the GNN implicitly solves the long-horizon scheduling problem over multiple slots.
\begin{figure}
    \centering
    \includegraphics[width=0.95\linewidth]{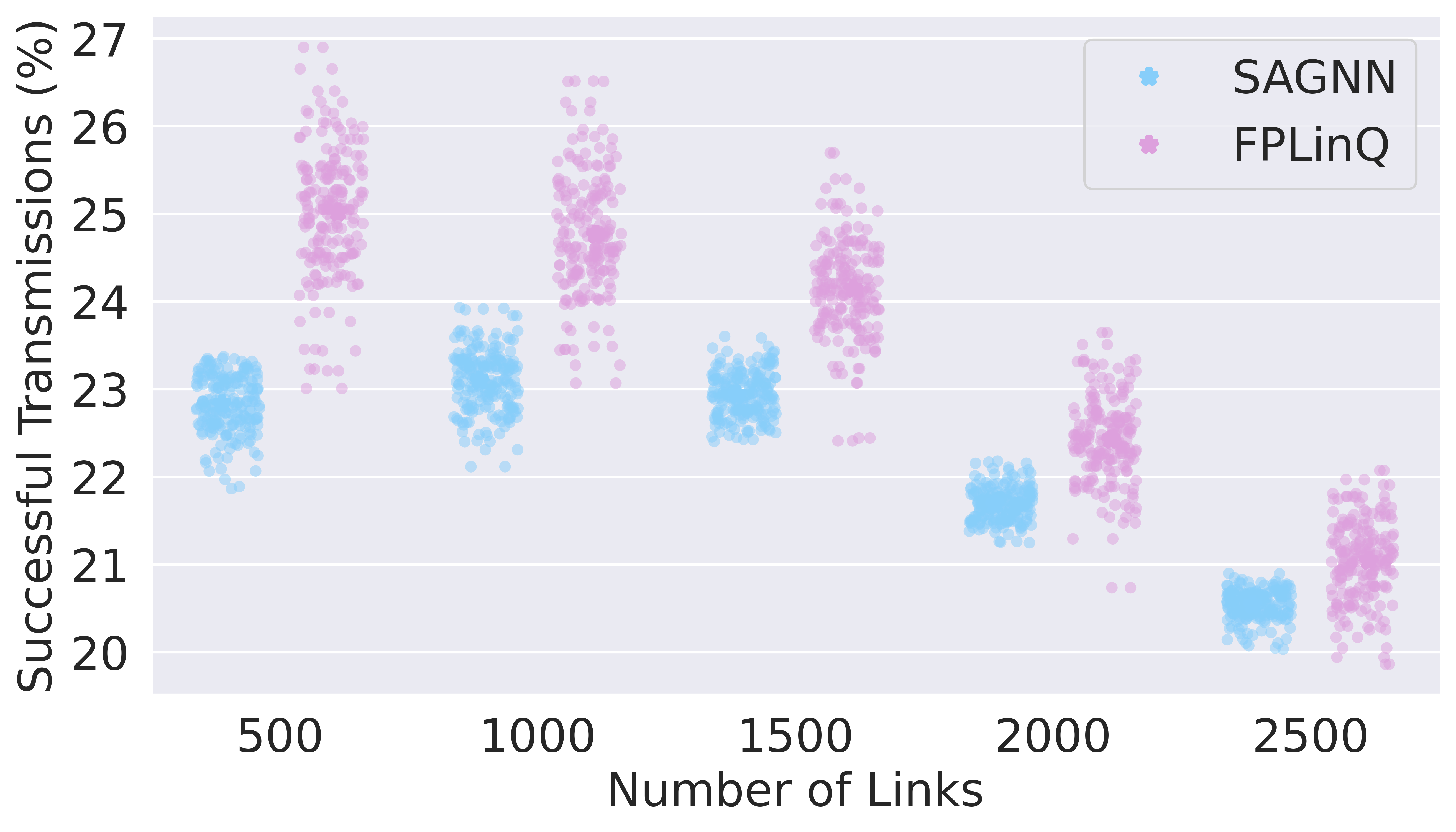}
    \caption{Performance of a GNN trained to solve WLS in networks with $c\simeq500$ links. The model is deployed on larger-scale networks and compared against the FPLinQ baseline.}
    \label{fig:scalabilitysawl}
\end{figure}

Figure \ref{fig:100vs500scaled} examines the impact of training graph scale on the per-link empirical rate distribution under a minimum long-term transmission constraint $\bbdelta=0.1$. Comparing models trained at $c_1=100$ and $c_2=500$, both exhibit stable scalability properties; however, the model trained on $c_2=500$ yields more high-rate allocations.
\begin{figure}
    \centering
    \includegraphics[width=0.95\linewidth]{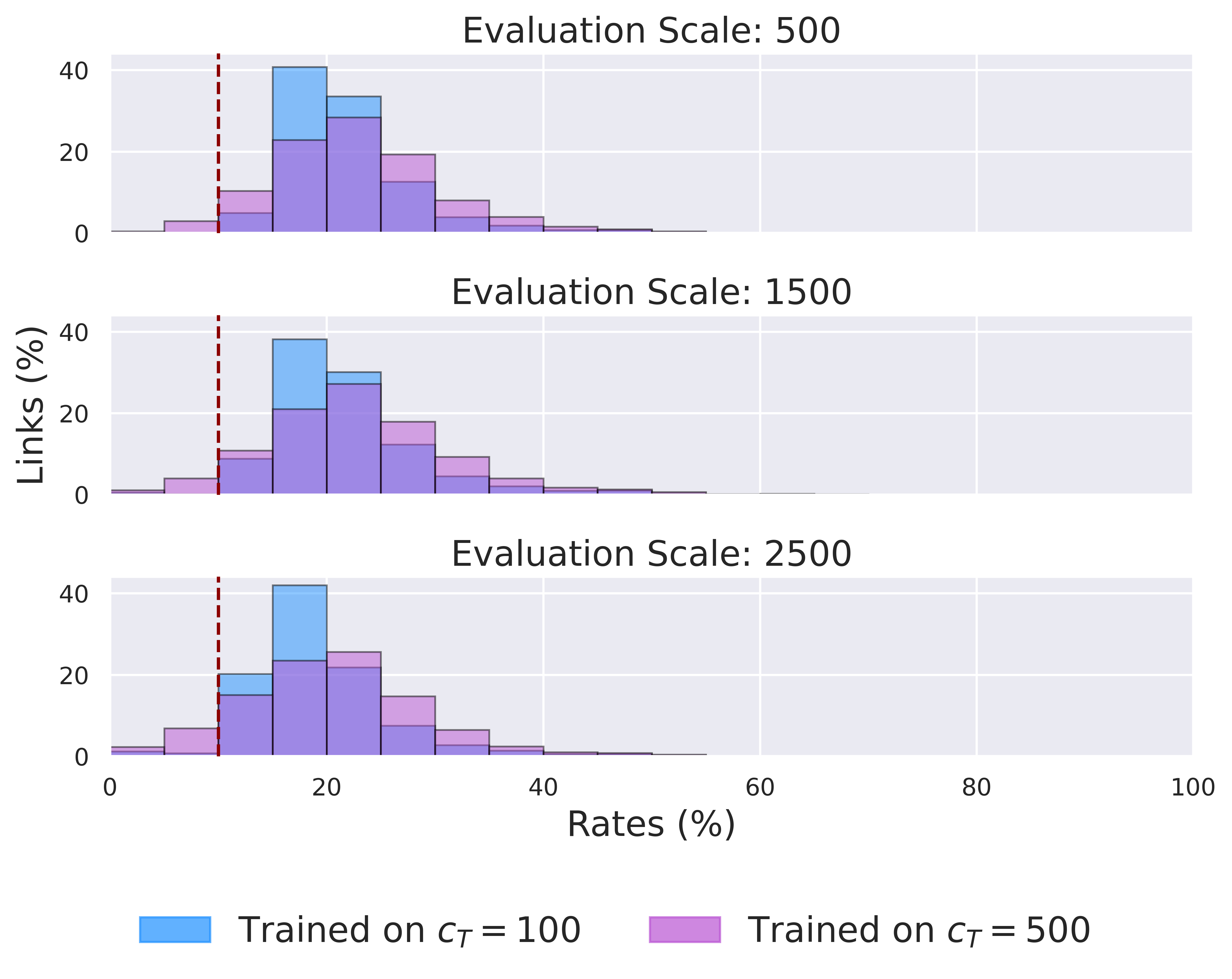}
    \caption{Distribution of links achieving different transmission rates. Two models are trained on scales $c_1=100$ and $c_2=500$.}
    \label{fig:100vs500scaled}
\end{figure}

\subsection{Robustness to different noise levels and relevance of closeness}
We evaluate the transferability of a GNN trained for WLS at scale $c=500$ across varying perturbations in Figure \ref{fig:stabilityovernoises}. Models trained on more regular topologies (low $\sigma$) exhibit worse generalization when deployed onto irregular graphs. Nonetheless, models trained on higher-noise topologies demonstrate robust backward compatibility, maintaining high performance when deployed onto near-grid structures.
\begin{figure}
    \centering
    \includegraphics[width=0.95\linewidth]{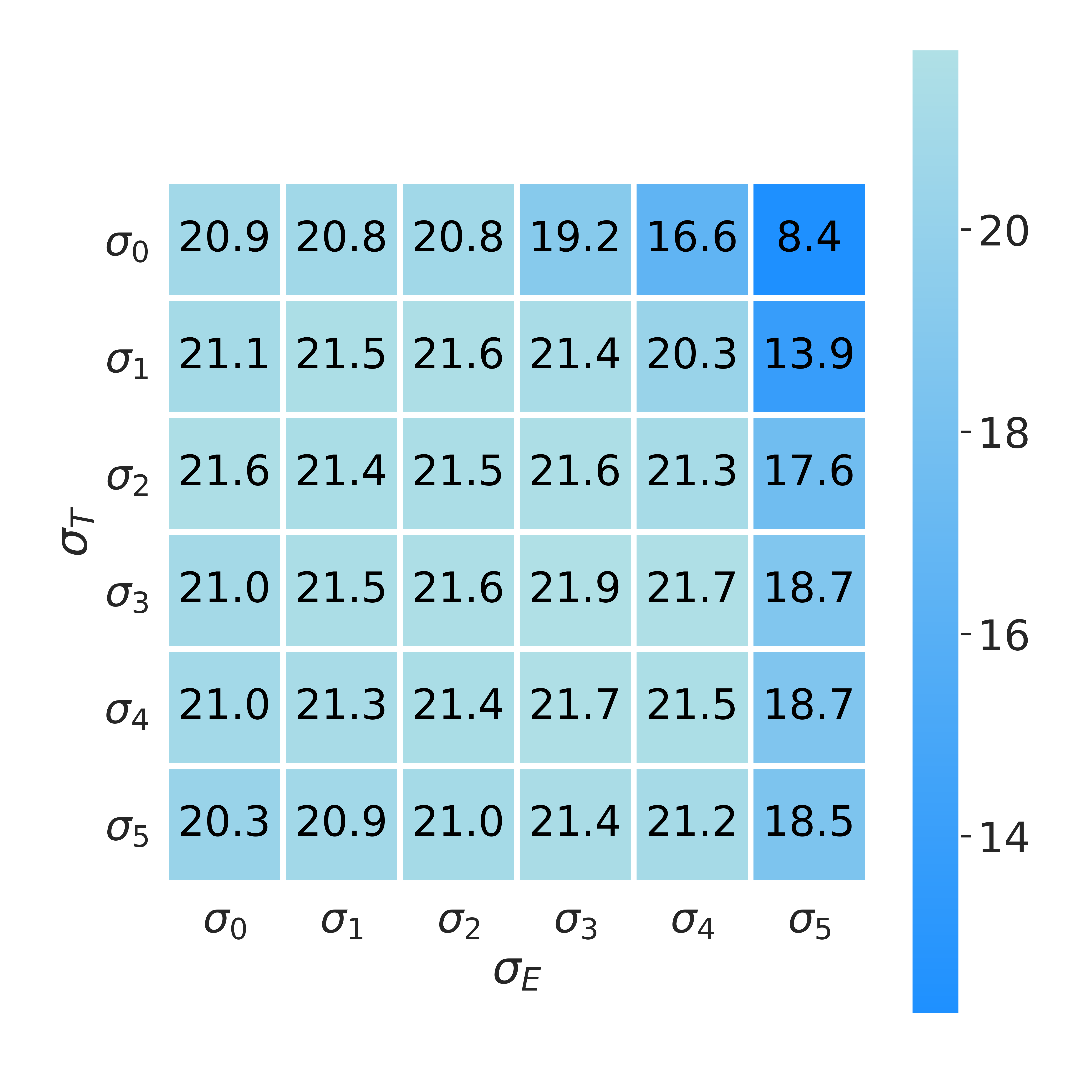}
    \caption{Results for the objective function in WLS as models are trained for different noise levels on the same scale, $c=500$.}
    \label{fig:stabilityovernoises}
\end{figure}

We further investigate the relevance of the closeness assumption of the RGGs to DGGs for both PA (Figure \ref{fig:relevanceofclosenessPA}) and WLS (Figure \ref{fig:relevanceofclosenessSAWL}). Models trained on RGGs closer to the underlying DGG exhibit better scalability, achieving higher sum rates alongside minimal constraint violations. As shown in Figure \ref{fig:relevanceofclosenessPA}, a GNN trained on unperturbed grid graphs scales seamlessly to larger network sizes, maximizing the sum rate while adhering to the total power budget $P_{\max}$. As the training noise level increases, performance gains over the baseline diminish, particularly on large-scale instances. While high-noise models exhibit lower sum-rate efficiency, their power constraint violations remain non-positive. From a practical standpoint, this indicates that while the policy under-utilizes available power, it ensures operational feasibility.

A similar trend holds for the WLS task (Figure \ref{fig:relevanceofclosenessSAWL}), where performance degrades as training noise increases. To ensure a meaningful comparison, we report the constraint violation as the empirical percentage of links failing to satisfy the minimum transmission probability $\bbdelta$. Because the FPLinQ benchmark optimizes an instantaneous formulation rather than long-term rate constraints, its constraint violations are inherently high, around $70-80\%$. 
\begin{figure}
    \centering
    \includegraphics[width=0.95\linewidth]{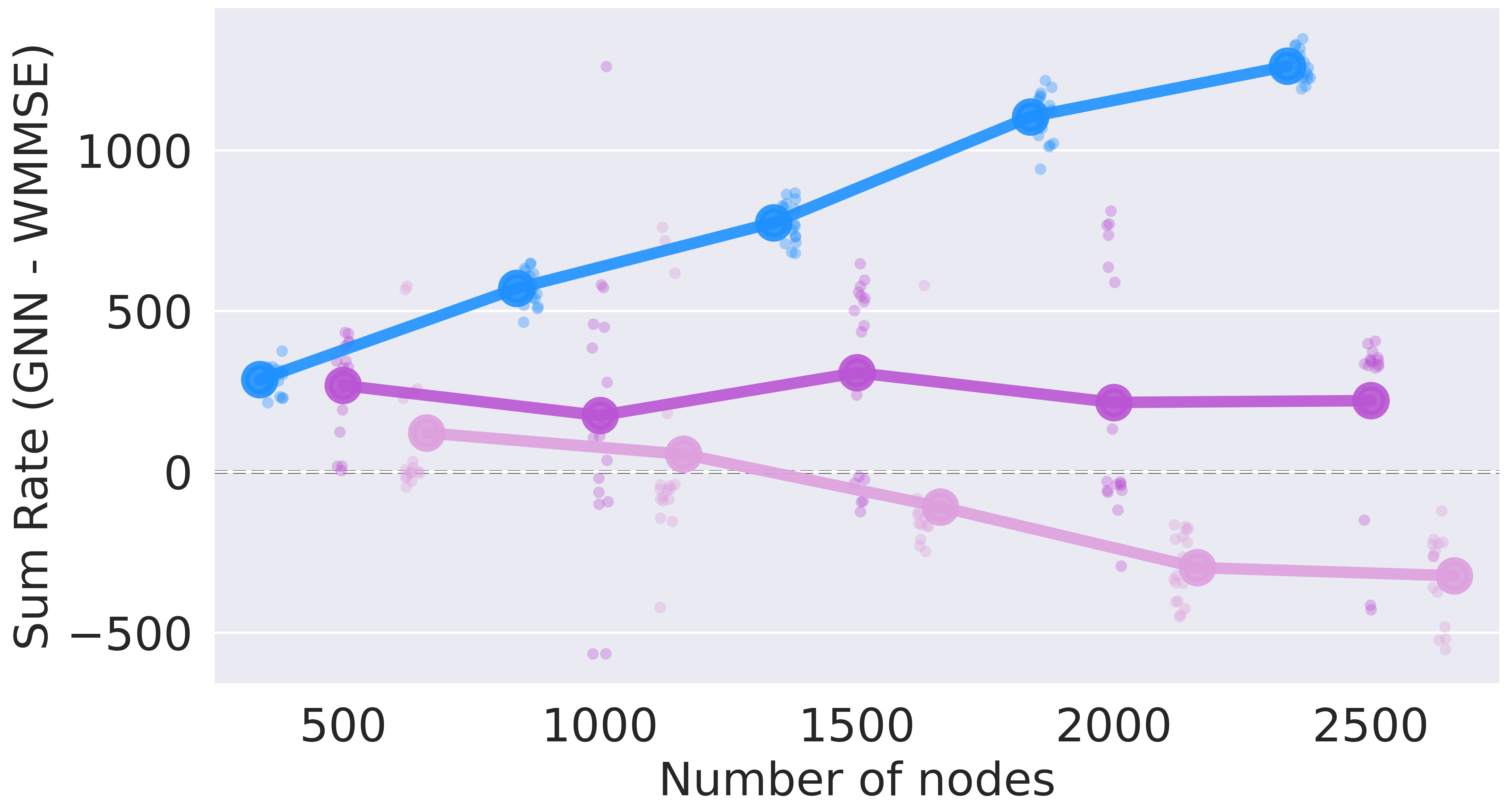}
    \includegraphics[width=0.95\linewidth]{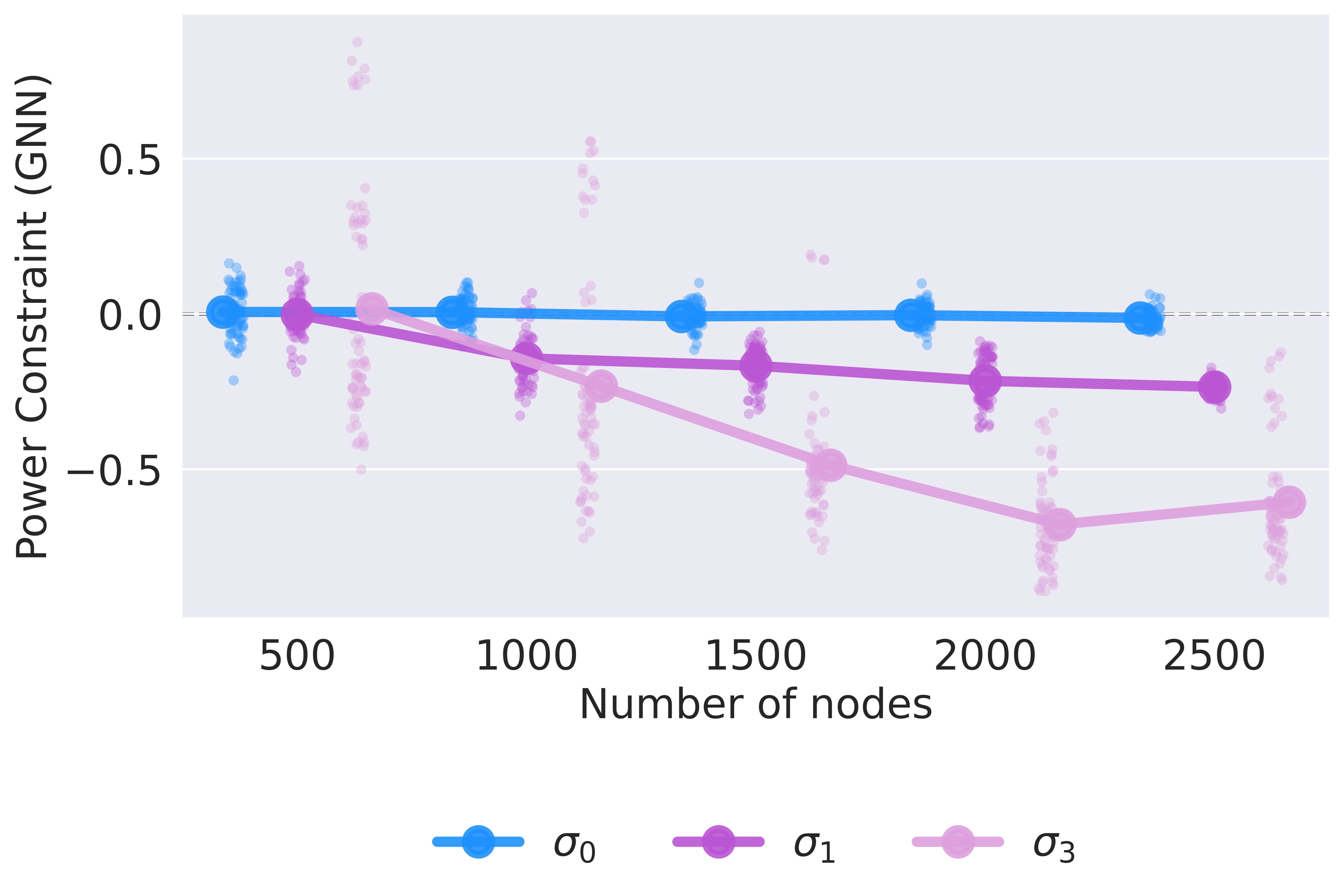}
    \caption{Difference in performance in comparison to the WMMSE baseline for a GNN trained to solve PA in scale $m=500$. We show the objective function (top) and the power constraint (bottom).}
    \label{fig:relevanceofclosenessPA}
\end{figure}

\begin{figure}
    \centering
    \includegraphics[width=0.95\linewidth]{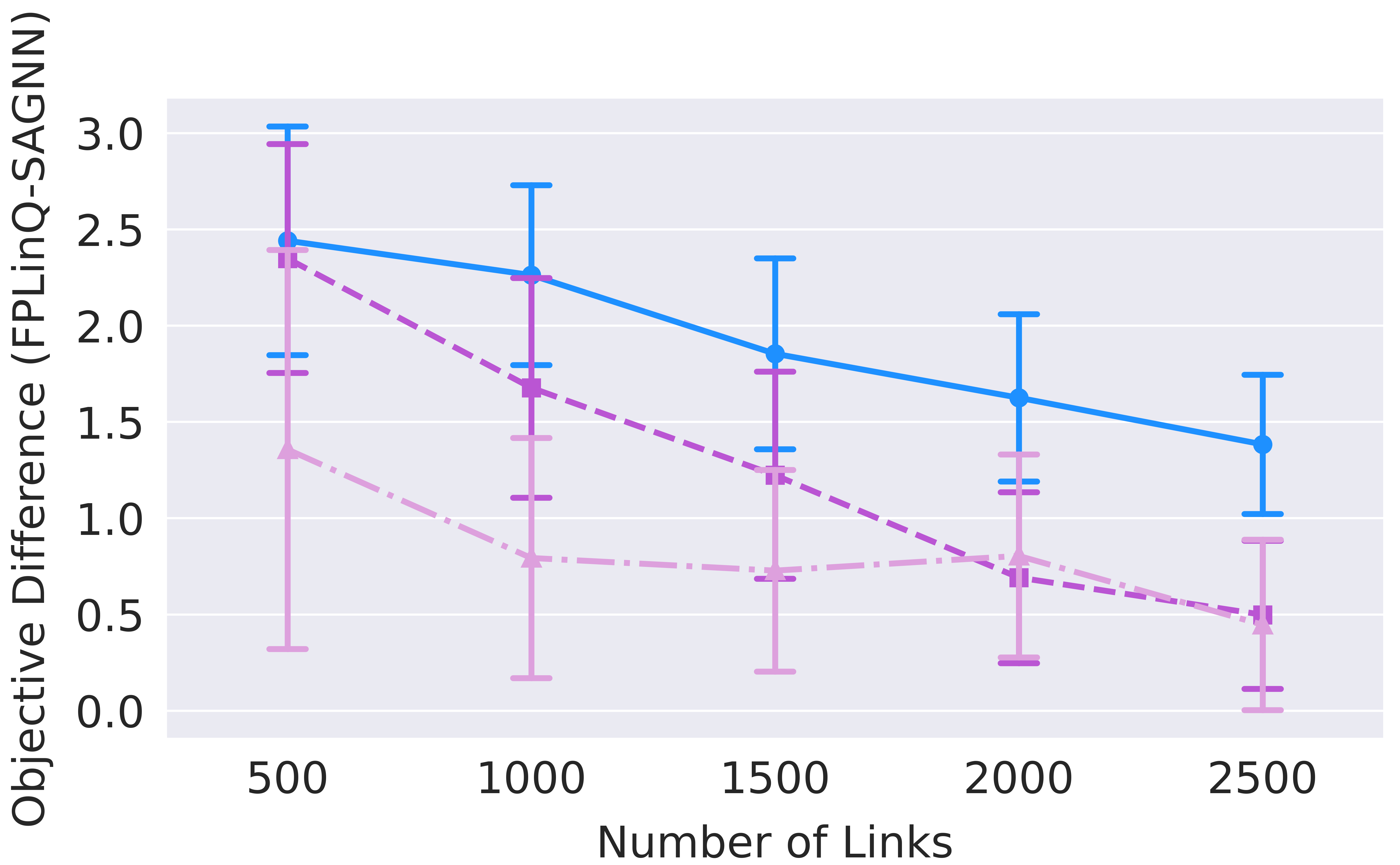}
    \includegraphics[width=0.95\linewidth]{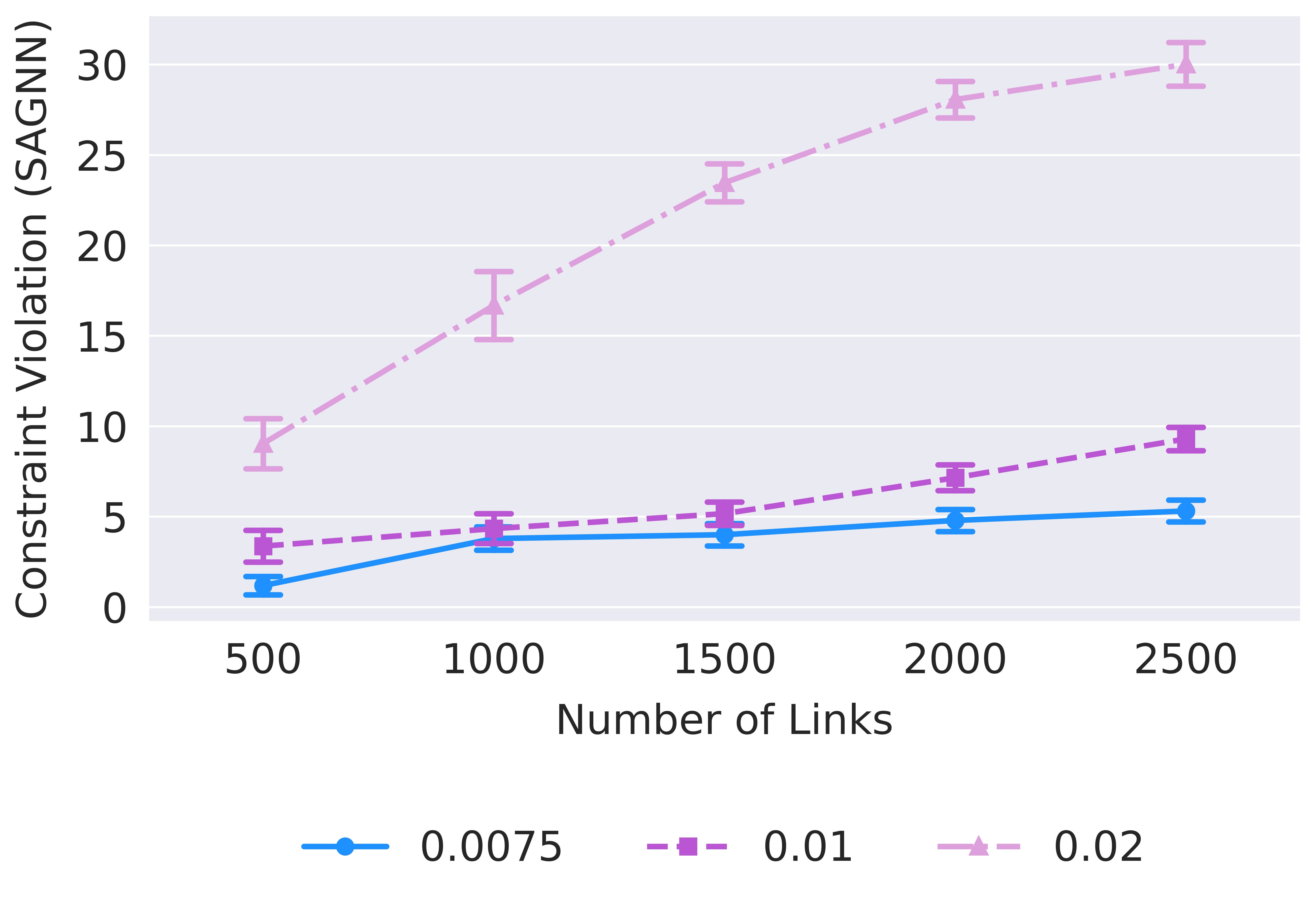}
    \caption{Relevance of closeness to DGGs for a GNN trained to solve WLS on RGGs of scale $c=500$. We show the difference in performance for the objective with respect to FPLinQ, and then the absolute constraint violation achieved by the GNN.}
    \label{fig:relevanceofclosenessSAWL}
\end{figure}

\subsection{Impact of the integral Lipschitz assumption}
Our theoretical results depend on Assumption \ref{ass:lipschitzfilter}, which requires graph filters to be integral Lipschitz. We verify the relevance of this assumption by comparing our default model with a GNN trained to learn integral Lipschitz filters with small constants. The training loss $\ccalL$ is penalized as follows:
\begin{align}
    \ccalL_{\text{Lipschitz}}&= \ccalL + \alpha \lambda_{\max},\nonumber\\
    \lambda_{\max}&=\max_\lambda \|\bm\lambda h'(\lambda)\|,
\end{align}
where $\alpha=5$ is a hyperparameter. 

The integral Lipschitz GNN (IL-GNN) and the regular GNN are trained on scale $m=500$ for the power allocation task, with evaluation shown in Table \ref{tab:ilgnnvsgnn}. As expected, the IL-GNN model scales better, achieving higher sum rates with power constraint values closer to zero, but the differences are not significant. 

\begin{table*}[]
    \centering
    \begin{tabular}{llllllll}
\textbf{Scale} &  & \textbf{500} & \textbf{1000} & \textbf{1500} & \textbf{2000} & \textbf{2500} & \textbf{3000} \\ \hline
\begin{tabular}[c]{@{}l@{}}\textbf{Sum} \\ \textbf{rate}\end{tabular} & \textbf{GNN} & $849.2\pm261.3$ & $760.0\pm375.3$ & $1814.4\pm658.5$ & $1844.9\pm1159.9$ & $3233.2\pm379.4$ & $2097.7\pm1362.9$ \\ \cline{2-8} 
 & \textbf{IL-GNN} & $884.5\pm268.3$ & $800.5\pm384.0$ & $1871.0\pm661.6$ & $1914.3\pm1149.7$ & $3250.3\pm408.2$ & $2141.7\pm1288.7$ \\ \hline
\begin{tabular}[c]{@{}l@{}}\textbf{Power} \\ \textbf{constraint}\end{tabular} & \textbf{GNN} & $-0.064\pm0.0067$ & $-0.190\pm0.071$ & $-0.195\pm0.070$ & $-0.244\pm0.098$ & $-0.234\pm0.070$ & $-0.278\pm0.145$ \\ \cline{2-8} 
 & \textbf{IL-GNN} & $-0.028\pm0.0055$ & $-0.074\pm0.062$ & $-0.117\pm0.048$ & $-0.061\pm0.157$ & $-0.122\pm0.159$ & $-0.013\pm0.232$ \\ \cline{2-8} 
\end{tabular}
    \caption{Scalability comparison between a GNN and a GNN with integral Lipschitz graph filters.}
    \label{tab:ilgnnvsgnn}
\end{table*}

%% file: conclusions.tex
This work establishes a theoretical framework proving the scale transferability of Graph Neural Networks for wireless resource allocation tasks. By modeling wireless topologies as sparse Random Geometric Graphs and formulating them as spatial perturbations of Deterministic Grid Graphs, the analysis captures distance-dependent signal attenuation while maintaining analytical tractability. Numerical experiments validate the derived theoretical bounds, demonstrating consistent scale generalization across both node- and edge-level tasks on RGGs and their conflict graphs. Furthermore, empirical evaluations confirm the validity and practical relevance of the theoretical assumptions.

Future research can extend this framework along several directions. Generalizing the spatial perturbation technique to broader classes of sparse random graphs and non-uniform spatial distributions will advance GNN scalability theory beyond wireless systems. Adapting these geometric bounds to alternative graph architectures, such as Graph Transformers, represents a natural theoretical extension. Finally, relaxing signal stationarity assumptions and incorporating dynamic, time-varying topologies will extend applicability to highly mobile wireless networks with non-stationary channels

%% file: appendix.tex
\subsection{Scalability of Graph Neural Networks over DGG}
\label{sec:appendixproofdggtodgg}
\begin{theorem}\label{th:dggonndggappendix}
     Let $\bbPhi(\bbx_n, \bbG_n; \ccalH_m)$ and $\bbPhi(\bbx_m, \bbG_m; \ccalH_m)$ denote the parameterized policies evaluated on DGGs with $n$ and $m$ nodes and normalized adjacency matrices. With Assumptions \ref{ass:windowingoperationDGGs}--\ref{ass:lipschitznonlinearity}, the scalability (cf. \eqref{eq:scalability}) is bounded as 
    \begin{align}
        \Delta(\bbG_m, \bbG_n, \ccalH_m) \leq \frac{H_{L,K}^2}{m} \left( 2\sqrt{m}LK + L^2K^2 \right) \bbx_\text{max}^2,
    \end{align}
    where $H_{L,K} = \max_{l \in \{1, \dots, L\}} \sum_{k=0}^{K} |h_{l, k}|$.
\end{theorem}

\begin{proof}
We proceed by induction over the GNN layers. For the base case $l=0$, the windowed input signals match exactly ($\sqcap_m \bbx_{n,0} = \bbx_{m,0}$), rendering the initial error zero. Assuming the hypothesis holds for layer $l$, we extend it to layer $l+1$. 

Since the activation function $\gamma$ is pointwise, it commutes with the windowing operator. Leveraging its normalized Lipschitz continuity allows us to drop the activation without expanding the norm:
\begin{align}
    \|&\sqcap_{m}(\bbx_{n, l+1})-\bbx_{m, l+1}\| \nonumber\\&= \bigg\|\gamma\left(\sqcap_{m}\sum_{k=0}^{K} h_{l+1, k}\bbG_n^k\bbx_{n,l}\right) - \gamma\left(\sum_{k=0}^{K} h_{l+1, k}\bbG_m^k\bbx_{m, l}\right)\bigg\| \nonumber\\
    &\leq \bigg\|\sqcap_{m}\sum_{k=0}^{K} h_{l+1, k}\bbG_n^k\bbx_{n,l} - \sum_{k=0}^{K} h_{l+1, k}\bbG_m^k\bbx_{m, l}\bigg\|. \nonumber\\
    \label{eq:induct_split1}
\end{align}

Applying the triangle inequality takes the sumation outside of the norm.
\begin{align}
    \|&\sqcap_{m}(\bbx_{n, l+1})-\bbx_{m, l+1}\| \nonumber\\&\leq \sum_{k=0}^{K} |h_{l+1, k}| \left\|\sqcap_{m}\bbG_n^k\bbx_{n,l} - \bbG_m^k\bbx_{m, l}\right\|. 
\end{align}
We add and subtract the term $\sqcap_{m}\bbG_n^k\sqcap_{m}\bbx_{m, l}$ to isolate the boundary truncation from the propagated error. Applying the triangle inequality yields:
\begin{align}
    \|\sqcap_{m}(&\bbx_{n, l+1}-\bbx_{m, l+1})\| \nonumber\\
    &\leq \sum_{k=0}^{K} |h_{l+1, k}| \Big( \left\|\sqcap_{m}\bbG_n^k (\bbx_{n,l} - \sqcap_m \bbx_{m, l})\right\| \nonumber \\ 
    &\quad \quad + \left\|(\sqcap_{m}\bbG_n^k\sqcap_{m}- \bbG_m^k)\bbx_{m, l}\right\| \Big). \label{eq:induct_split2}
\end{align}

To bound the right-hand side of Equation \eqref{eq:induct_split2}, we observe the first term. Using subadditivity and knowing the adjacency matrix is normalized $\|\bbG_n\|_2 \leq 1$, the first term splits into the accumulated error inside the window and the external signal leaking in from the $k$-hop neighborhood:
\begin{align}\label{eq:firsttermdgg}
    \|\sqcap_{m}\bbG_n^k (&\bbx_{n,l} - \sqcap_m \bbx_{m, l})\| \nonumber \\
    &\leq \|\sqcap_m (\bbx_{n,l} - \bbx_{m,l})\| + \|(\sqcap_{m+k} - \sqcap_m)\bbx_{n,l}\|.
\end{align}
The first component in Equation \eqref{eq:firsttermdgg} is the accumulated error at layer $l$, which is bounded by the inductive hypothesis. The second component represents the signal restricted to nodes outside the window $m$ but within its $k$-hop neighborhood.

For the second term in Equation \eqref{eq:induct_split2}, full-graph diffusion followed by truncation matches a local $k$-hop diffusion on the truncated graph everywhere except near the perimeter. This term evaluates to non-zero only for interior boundary nodes located less than $k$ hops from the graph's edge:
\begin{align}
    \|(\sqcap_{m}\bbG_n^k\sqcap_{m}- \bbG_m^k)\bbx_{m, l}\| \leq \|(\sqcap_{m} - \sqcap_{m-k})\bbx_{m,l}\|.
\end{align}

We can now combine these terms to bound the total error at layer $l+1$. Let $e_l = \|\sqcap_m (\bbx_{n,l} - \bbx_{m,l})\|$. Substituting our bounded terms back into Equation \eqref{eq:induct_split2} yields:
\begin{align}
    e_{l+1} \leq& \sum_{k=0}^{K} |h_{l+1, k}| \Big( e_l + \|(\sqcap_{m+k} - \sqcap_m)\bbx_{n,l}\| \nonumber\\
    &\quad+ \|(\sqcap_{m} - \sqcap_{m-k})\bbx_{m,l}\| \Big).
\end{align}

The two rightmost terms represent the signal restricted to nodes within a $k$-hop radius of the boundary of $m$. By the structural properties of a DGG, the total number of nodes in this boundary region is bounded by $2\sqrt{m}k + k^2$. Because the signal amplitude is bounded by $\bbx_{\text{max}}$, the Euclidean norm of the signal over these boundary nodes is at most $\sqrt{2\sqrt{m}k + k^2}\bbx_{\text{max}}$.

Unrolling this recursion from $l=1$ to $L$ propagates the error outward. Since the graph filters expand the localized neighborhood by at most $K$ hops per layer, after $L$ layers, the total accumulated error is strictly confined to nodes within $LK$ hops of the boundary. The maximum number of affected boundary nodes expands to $2\sqrt{m}LK + L^2K^2$. The signal amplitude propagated through these layers is scaled by the maximum filter coefficient sum, bounded by $H_{L,K} = \max_{l \in \{1, \dots, L\}} \sum_{k=0}^{K} |h_{l, k}|$.

Therefore, the norm of the final error vector at layer $L$ is bounded by the maximum amplitude applied over the affected boundary nodes:
\begin{align}
    e_L = \|\sqcap_m (\bbx_{n, L} - \bbx_{m, L})\| \leq H_{L,K} \sqrt{2\sqrt{m}LK + L^2K^2} \bbx_{\text{max}}.
\end{align}

We square the norm of the final error and divide by the number of nodes $m$, obtaining the final bound:
\begin{align}
    \Delta(\bbG_m, \bbG_n, \ccalH_m) &= \frac{1}{m} \| \sqcap_m (\bbx_{n, L} - \bbx_{m, L}) \|^2 \nonumber \\
    &\leq \frac{H_{L,K}^2}{m} \left( 2\sqrt{m}LK + L^2K^2 \right) \bbx_{\text{max}}^2.
\end{align}
\end{proof}

\subsection{Scalability of Graph Neural Networks between RGG and DGG}
\label{sec:appendixproofrggtodgg}

\begin{theorem}\label{the:rgg-gg-gnntransfappendix}
    Let $\bbR_m$ and $\bbG_m$ be an RGG and a DGG with $m$ nodes, with normalized adjacency matrices. A GNN $\bbPhi(\cdot, \cdot; \ccalH_m)$ is trained with parameters $\ccalH_m$, under Assumptions \ref{ass:closeness}, \ref{ass:stationarity}--\ref{ass:lipschitzfilter}. The degradation in performance when transferring the model from RGGs to DGGs is bounded,
    \begin{equation}\label{eq:boundrggtodgg-app}
        \Tilde{\Delta}(\bbR_{m}, \bbG_{m}, \ccalH_m) \leq 4L^2\varepsilon C\bbx_{\text{max}}^2.
    \end{equation}
\end{theorem}
\begin{proof}
We begin by showing the transferability over scales of graph filters. The adjacency matrices of the graphs are symmetric and therefore allow their eigendecomposition to be written. As the difference between the matrices is $\|\bbG_m-\bbR_m\|\leq \varepsilon$, the difference between their eigenvalues is also bounded:
\begin{align}
    |\tilde{\lambda}_i-\lambda_i|\leq \varepsilon
\end{align}
Let $\delta>0$ small. For a given eigenvalue $\tilde{\lambda}_i$ of $\bbG_m$, we have eigenvalues ${\lambda}_{j}$ of $\bbR_m$ that are closer than $\delta$, such that $|\tilde{\lambda}_{i}-\lambda_j|\leq \delta$, and eigenvalues ${\lambda}_{j}$ of $\bbR_m$ that are further away than $\delta$, such that $|\tilde{\lambda}_{i}-\lambda_j|> \delta$. Let us use this to construct an eigenbasis via the eigenvalues of $\bbR_m$ as follows.
\begin{align}
    \bbV_i =& [{v}_j]_{j:|\tilde{\lambda}_{i}-\lambda_j|\leq \delta}\nonumber\\
    \bbV_i^C =& [{v}_j]_{j:|\tilde{\lambda}_{i}-\lambda_j|> \delta}
\end{align}
Let $\tilde{v_i}$ be an eigenvector of $\bbR_m$, we can write $\tilde{v_i} = \bbV_i\bbV_i^H\tilde{v_i}+\bbV_i^C(\bbV_i^{C})^H\tilde{v_i}$. We apply the graph filter to both graphs with input $\tilde{v_i}$ and compare the corresponding outputs. 
\begin{align}
    \bbh_m(\bbG_m)\tilde{v_i}-\bbh_m(\bbR_m)\tilde{v_i} = (\bbh_m(\bbG_m)-\bbh_m(\bbR_m))\tilde{v_i} = \Delta_i
\end{align}
Because the eigenvectors form a complete orthonormal basis, bounding the filter output difference for a given eigenvector $\tilde{v_i}$ is sufficient to bound the difference in the outputs for any arbitrary graph signal. Next, we consider the squared operator norm of $\Delta_i$. We replace the filter applied to $\bbG_m$, $\bbh_m(\bbG_m)$ by the frequency response representation of the filter.
\begin{align}\label{eq:delta_i}
    \|\Delta_i\|^2 =& \|h(\tilde{\lambda}_{i})\tilde{v_i}-\bbh_m(\bbR_m)\tilde{v_i}\|^2\nonumber\\
    =&\|h(\tilde{\lambda}_{i})\bbI-\bbh_m(\bbR_m)\tilde{v_i}](\bbV_i\bbV_i^H\tilde{v_i}+\bbV_i^C(\bbV_i^{C})^H\tilde{v_i})\|^2\nonumber\\
    =&\|[h(\tilde{\lambda}_{i})\bbI-\bbh_m(\bbR_m)\tilde{v_i}]\bbV_i\bbV_i^H\tilde{v_i}\|^2 \nonumber\\
    &+ \|[h(\tilde{\lambda}_{i})\bbI-\bbh_m(\bbR_m)\tilde{v_i}]\bbV_i^C(\bbV_i^{C})^H\tilde{v_i}\|^2,
\end{align}
where the crossterms disappear because the vectors are orthogonal. Define the terms of Equation \eqref{eq:delta_i} $\Delta_i^1$ and $\Delta_i^2$, respectively. We have that $\Delta_i^1$ is associated to those eigenvalues of $\bbR_m$ that are closer than $\delta$ to the eigenvalue $\tilde{\lambda}_{i}$ of $\bbG_m$, while $\Delta_i^2$ considers those further away.

To bound $\Delta_i^1$ we expand $\bbh_m(\bbR_m)$ via the eigendecomposition of $\bbR_m$.
\begin{align}
\Delta_i^1 &= \bbh_m(\bbR_m)\bbV_i\bbV_i^H \tilde{v_i} - h(\tilde{\lambda_i})\bbV_i\bbV_i^H \tilde{v_i} \nonumber\\
&= \bbV h({\Lambda}) \bbV^H \bbV_i\bbV_i^H \tilde{v_i} - \bbV_i h(\tilde{\lambda_i})I \bbV_i^H \tilde{v_i} \nonumber\\
&= \bbV_i \left( h({\Lambda}_i) - h(\tilde{\lambda_i})\bbI \right) \bbV_i^H \tilde{v_i}
\end{align}
Let $h(\cdot)$ be an integral Lipschitz filter with constant $C$, we can bound $\Delta_i^1$ as follows.
\begin{align}
    \|\Delta_i^1\|^2 \leq C^2\delta^2,
\end{align}
where we consider the eigenvectors are orthonormal.

To bound $\Delta_i^2$, we apply the triangle inequality to the operator norm of the filter difference. Noting the filter coefficients are normalized such that $\|\bbh_m(\bbR_m)\| = \max_j |h({\lambda}_j)| \leq 1$, we get:
\begin{align}\label{eq:firsttermdelta2}
    \|\bbh_m(\bbR_m)-h(\tilde{\lambda_i})\bbI\|^2\leq& (\|\bbh_m(\bbR_m)\|+\|h(\tilde{\lambda_i})\bbI\|)^2\nonumber\\
    \leq&(1 + 1)^2 = 4.
\end{align}
We must also bound the projection of $\tilde{v_i}$ onto the complement subspace, $\|(\bbV_i^C)^H\tilde{v_i}\|^2$. Using the fact that $(\bbR_m - \bbG_m)\tilde{v_i} = (\bbR_m - \tilde{\lambda_i} \bbI)\tilde{v_i}$, we left-multiply by $(\bbV_i^C)^H$ to obtain:
\begin{align}
    (\bbV_i^C)^H(\bbR_m-\bbG_m)\tilde{v_i} =& ({\Lambda}_i^C - \tilde{\lambda_i} \bbI)(\bbV_i^C)^H\tilde{v_i}
\end{align}

Taking the norm of both sides and applying the Davis-Kahan theorem, we bound the vector projection. Since the eigenvalues in ${\Lambda}_i^C$ are strictly further than $\delta$ from $\tilde{\lambda_i}$, we have:
\begin{align}\label{eq:secondtermdelta2}
    \|(\bbV_i^C)^H\tilde{v_i}\| \leq \frac{\|(\bbR_m-\bbG_m)\tilde{v_i}\|}{\min_j |\tilde{\lambda_i} - {\lambda}_j^C|} \leq \frac{\varepsilon}{\delta}.
\end{align}
Putting Equations \eqref{eq:firsttermdelta2} and \eqref{eq:secondtermdelta2} together, we obtain a bound for $\Delta_i^2$:
\begin{align}
    \|\Delta_i^2\|^2 \leq& \|\bbh_m(\bbR_m)-h(\tilde{\lambda_i})\bbI\|^2 \|(\bbV_i^C)^H\tilde{v_i}\|^2 \nonumber\\
    \leq& 4\frac{\varepsilon^2}{\delta^2}
\end{align}

Finally, substituting the bounds for both terms yields the final bound for the squared error:
\begin{align}\label{eq:prefinalbound}
    \|\Delta_i\|^2 =& \|\bbh_m(\bbR_m)\tilde{v_i} - h(\tilde{\lambda_i})\tilde{v_i}\|^2\nonumber\\
    \leq& C^2\delta^2 + 4\frac{\varepsilon^2}{\delta^2}.
\end{align}
Let $\bbx\in\mathbb{R}^{m}$ be the input to the graph filter, we can write $\bbx$ in terms of the eigenvectors $\tilde{v_i}$:
\begin{align}
    \bbx = \sum_{i=1}^n (\tilde{v_i}^H\bbx \tilde{v_i}) = \sum_{i=1}^n (\mathbf{\tdx}_i \tilde{v_i}).
\end{align}
To finish the proof, we bound the squared norm of $\Delta_i\bbx$ by observing the eigenvectors are orthogonal and therefore the cross products are zero.
\begin{align}
    \|\Delta_i\bbx\|_2^2 =& \|\Delta_i\sum_{i=1}^n (\mathbf{\tdx}_i \tilde{v_i})\|_2^2\nonumber \\
    =& \sum_{i=1}^n \mathbf{\tdx}_i^2\|\Delta_i \tilde{v_i}\|_2^2,
\end{align}
where we find that the first term is $\|\bbx\|_2^2$ and the second term is what we bounded in Equation \eqref{eq:prefinalbound}. Finally, if we define $\delta^2=\frac{2\varepsilon}{C}$, we get the final bound:
\begin{align}
    \|\bbh_m(\bbR_m)\bbx-\bbh_m(\bbG_m)\bbx\|^2 \leq 4C\varepsilon \|\bbx\|^2.
\end{align}

Until now we have shown transferability of graph filters. To extend these results to GNNs, we compare its outputs applied to both graphs via $\bbR_m$ and $\bbG_m$ for layer $l$, considering the pointwise nonlinearity $\sigma(\cdot)$ is normalized Lipschitz.
\begin{align}
    \|\bbPhi_l(\bbx, \bbG_m; \bbh_m)-\bbPhi_{l}&(\bbx, \bbR_m; \bbh_m)\| \nonumber\\
    &=\|\sigma(\bbh_{m_{l}}(\bbG_m)\bbPhi_{l-1}(\bbx, \bbG_m; \bbh_m))\nonumber\\
    \quad &-\sigma(\bbh_{m_{l}}(\bbR_m)\bbPhi_{l-1}(\bbx, \bbR_m; \bbh_m))\| \nonumber\\
    \leq& \|\bbh_{m_{l}}(\bbG_m)\bbPhi_{l-1}(\bbx, \bbG_m; \bbh_m)\nonumber \\
    \quad &-\bbh_{m_{l}}(\bbR_m)\bbPhi_{l-1}(\bbx, \bbR_m; \bbh_m)\|.
\end{align}
Add and subtract the cross term $\bbh_{m_l}(\bbG_m)\bbPhi_{l-1}(\bbx, \bbR_m; \bbh_m)$. By applying the triangle inequality, we establish a recursion:
\begin{align}
    \|\bbPhi_l(\bbx, \bbG_m; \bbh_m)-\bbPhi_{l}&(\bbx, \bbR_m; \bbh_m)\| \nonumber\\
    &\leq \|\bbh_{m_{l}}(\bbG_m)\bbPhi_{l-1}(\bbx, \bbG_m; \bbh_m)\nonumber\\
    \quad&-\bbh_{m_{l}}(\bbG_m)\bbPhi_{l-1}(\bbx, \bbR_m; \bbh_m)\|\nonumber\\
    &+\|\bbh_{m_{l}}(\bbG_m)\bbPhi_{l-1}(\bbx, \bbR_m; \bbh_m)\nonumber\\
    \quad&-\bbh_{m_{l}}(\bbR_m)\bbPhi_{l-1}(\bbx, \bbR_m; \bbh_m)\|\nonumber \\
    \leq& \|\bbh_{m_{l}}(\bbG_m)\|\|\bbPhi_{l-1}(\bbx, \bbG_m; \bbh_m)\nonumber\\
    \quad&- \bbPhi_{l-1}(\bbx, \bbR_m; \bbh_m)\|\nonumber \\
    &+\|\bbh_{m_{l}}(\bbG_m)-\bbh_{m_{l}}(\bbR_m)\|\nonumber\\
    \quad\quad&\|\bbPhi_{l-1}(\bbx, \bbR_m; \bbh_m)\|.
\end{align}
We assume that the filter coefficients are normalized, such that $\|\bbh_m{m_l}(\bbG_m)\|\leq 1$. Furthermore, the combination of non-amplifying filters and the normalized Lipschitz property of the activation function ensures that $\|\bbPhi_{l-1}(\bbx, \bbR_m; \bbh_m)\|\leq\|\bbx\|$. Incorporating the transferability bound of the single graph filter, we obtain:
\begin{align}
    \|\bbPhi_l(\bbx, \bbG_m; &\bbh_m)-\bbPhi_{l}(\bbx, \bbR_m; \bbh_m)\| \leq \nonumber\\
    &\|\bbPhi_{l-1}(\bbx, \bbG_m; \bbh_m)- \bbPhi_{l-1}(\bbx, \bbR_m; \bbh_m)\| \nonumber\\
    &\quad+2\sqrt{\varepsilon C}\|\bbx\|.
\end{align}
By unrolling the recursion over the layers, we arrive at the final bound for the outputs of a GNN with $L$ layers.
\begin{align}
    \|\bbPhi_L(\bbx, \bbG_m; \bbh_m)-\bbPhi_L(\bbx, \bbR_m; \bbh_m)\| \leq& 2L\sqrt{\varepsilon C}\|\bbx\|.
\end{align}
To align our result with the definition of scalability, we square both sides of the inequality and divide by the number of nodes $m$. By bounding the squared Euclidean norm of the $m$-dimensional input signal with its maximum absolute value, we apply the inequality $\|\bbx\|_2^2\leq m\|\bbx\|_{\infty}^2=\bbx_{\max}^2$. The $m$ in the denominator cancels out, yielding the final bound:
\begin{align}
    \Tilde{\Delta}(\bbR_{m}, \bbG_{m}, \ccalH_m) \leq 4L^2\varepsilon C\bbx_{\text{max}}^2.
\end{align}
\end{proof}

\subsection{Scalability of Graph Neural Networks over RGG}
\label{sec:appendixproofrggtorgg}
\begin{theorem}\label{the:rggtonrgg-app}
    Under Assumptions \ref{ass:closeness}--\ref{ass:windowingoperationRGGs}, a GNN $\bbPhi$ is trained on an RGG of size $m$, with parameters $\ccalH_m$. 
    \begin{align}
         \Delta(\bbR_m, \bbR_n, \ccalH_m) \leq& 24L^2\varepsilon C \bbx_{\text{max}}^2\nonumber\\
    &+\frac{3H_{K,L}^2}{m}(2\sqrt{m}LK+L^2K^2)\bbx_{\max}^2
    \end{align}
\end{theorem}

\begin{proof}
We write the difference $\Delta(\bbR_m, \bbR_n, \ccalH_m)$ in its expanded form and apply the triangle inequality.
\begin{align}
    \|\sqcap_m(\bbPhi(\bbx_n&, \bbR_n; \ccalH))-\bbPhi(\bbx_m, \bbR_m;\ccalH)\| \nonumber\\
    \leq& \|\bbPhi(\bbx_m, \bbG_m;\ccalH)-\bbPhi(\bbx_m, \bbR_m;\ccalH)\| \nonumber\\
    &+\|\sqcap_m(\bbPhi(\bbx_n, \bbG_n;\ccalH))-\bbPhi(\bbx_m, \bbG_m;\ccalH)\|\nonumber\\
    &+ \|\sqcap_m(\bbPhi(\bbx_n, \bbG_n;\ccalH)-\bbPhi(\bbx_n, \bbR_n;\ccalH))\|
\end{align}
We square both sides and divide over $m$. We apply the Cauchy-Schwarz inequality for three terms, and consider stationarity of the signal.
\begin{align}
    \frac{1}{m}\|\sqcap_m(\bbPhi(\bbx_n,& \bbR_n; \ccalH))-\bbPhi(\bbx_m, \bbR_m;\ccalH)\|^2 \nonumber\\
    \leq& 3[4L^2\varepsilon C \bbx_{\text{max}}^2\nonumber\\
    &+\frac{H_{L,K}^2}{m}(2\sqrt{m}LK+L^2K^2)\bbx_{\text{max}}^2\nonumber\\
    &+4L^2\varepsilon C \bbx_{\text{max}}^2]\nonumber\\
    \leq& 24L^2\varepsilon C \bbx_{\text{max}}^2\nonumber\\
    &+\frac{3H_{K,L}^2}{m}(2\sqrt{m}LK+L^2K^2)\bbx_{\text{\text{max}}}^2
\end{align}
\end{proof}

\subsection{Scalability of Graph Neural Networks over conflict DGG}
\label{sec:appendixproofdggtodgg-conflict}
\begin{theorem}\label{th:dggonndgg-conflicts-appendix}
    Let $\bbPhi(\bbx_c, \hat{\bbG}_c; \ccalH_c)$ and $\bbPhi(\bbx_d, \hat{\bbG}_d; \ccalH_c)$ denote the parameterized policies evaluated on the conflict DGGs with $c$ and $d$ nodes and normalized adjacency matrices. Under Assumptions \ref{ass:stationarity}--\ref{ass:lipschitznonlinearity}, the scalability is bounded as
    \begin{align}\label{eq:boundfordggs-appendix}
        \Delta(\hat{\bbG}_c, \hat{\bbG}_d, \ccalH_c) \leq \frac{H_{L,K}^2}{c} \left( 4\sqrt{2c}LK + L^2K^2 \right) \bbx_\text{max}^2.
    \end{align} 
\end{theorem}
\begin{proof}
    The proof follows the exact inductive structure over layers $l \in \{1, \dots, L\}$ established in Appendix~\ref{sec:appendixproofdggtodgg}. The unrolling of graph convolutions and filter bounds remains unchanged. The sole modification lies in the boundary node count: whereas a standard $m$-node DGG perimeter scales as $2\sqrt{m}$, the conflict graph $\hat{\bbG}_c$ possesses an expanded boundary thickness of $2$, yielding $4\sqrt{2c}$ boundary nodes affected within a $1$-hop neighborhood. Propagating this boundary expansion across $LK$ hops yields at most $4\sqrt{2c}LK + L^2K^2$ affected boundary nodes. Substituting this count into the energy bound yields Equation~\eqref{eq:boundfordggs-appendix}.
\end{proof}

\subsection{Scalability of Graph Neural Networks between conflict RGG and DGG}
\label{sec:appendixproofrggtodgg-conflict}
\begin{theorem}\label{the:rggtondgg-conflict-appendix}
    Under Assumptions \ref{ass:closeness}, \ref{ass:stationarity}--\ref{ass:windowingoperationRGGs}, let a GNN $\bbPhi(\cdot, \cdot; \ccalH_c)$ be trained on conflict graphs of RGGs with $c$ nodes, yielding parameters $\ccalH_c$. When deployed without retraining onto a larger conflict graph $\hat{\bbR}_d$ of scale $d > c$, the scalability error is bounded as
    \begin{align}\label{eq:boundrggtonrgg-conflict-appendix}
        \Delta(\hat{\bbR}_c, \hat{\bbR}_d, \ccalH_c) \leq& 24L^2\varepsilon C \bbx_{\text{max}}^2 \\&+ \frac{3H_{L,K}^2}{c}\left(4\sqrt{2c}LK + L^2K^2\right)\bbx_{\text{max}}^2.
    \end{align}
\end{theorem}
\begin{proof}
    By the spectral equivalence of the signless edge and vertex Laplacians, the structural closeness $\|\bbG_m - \bbR_m\| \leq \varepsilon$ implies $\|\hat{\bbG}_c - \hat{\bbR}_c\| \leq \varepsilon$. With the spectral distance bounded, the spectral decomposition, Davis-Kahan projection, and integral Lipschitz filter bounds follow identically to the proof of Theorem~\ref{the:rgg-gg-gnntransf} in Appendix~\ref{sec:appendixproofrggtodgg}.
\end{proof}

\subsection{Scalability of Graph Neural Networks over conflict RGG}
\label{sec:appendixproofrggtorgg-conflict}
\begin{theorem}\label{the:rggtonrgg-conflict-appendix}
    Under Assumptions \ref{ass:closeness}--\ref{ass:windowingoperationRGGs}, let a GNN $\bbPhi(\cdot, \cdot; \ccalH_c)$ be trained on conflict graphs of RGGs with $c$ nodes, yielding parameters $\ccalH_c$. When deployed without retraining onto a larger conflict graph $\hat{\bbR}_d$ of scale $d > c$, the scalability error is bounded as
    \begin{align}\label{eq:boundrggtonrgg-conflict-appendix}
        \Delta(\hat{\bbR}_c, \hat{\bbR}_d, \ccalH_c) \leq& 24L^2\varepsilon C \bbx_{\text{max}}^2 \\&+ \frac{3H_{L,K}^2}{c}\left(4\sqrt{2c}LK + L^2K^2\right)\bbx_{\text{max}}^2.
    \end{align}
\end{theorem}
\begin{proof}
    Combining the results of Theorems~\ref{th:dggonndgg-conflicts} and \ref{the:rgg-gg-gnntransf-conflict} via the triangle inequality and Cauchy-Schwarz expansion follows the precise steps of Appendix~\ref{sec:appendixproofrggtorgg}, substituting the updated conflict DGG bound from Theorem~\ref{th:dggonndgg-conflicts}.
\end{proof}

\subsection{Experimental Details}
\label{sec:expdetails}
We present in Figure \ref{fig:parggsillustration} examples of RGGs perturbated with growing positional noise. In particular, the edges are weighted according to the channel interference used to run experiments for power allocation, with $\eta=1e-2$. The noise levels required for perturbations to affect the scalability in the original RGG graphs were experimentally found to be smaller in comparison to the noise levels in Figure \ref{fig:rggsandconflicts}, which were used for WLS. 

\begin{figure*}
    \centering
    \includegraphics[width=0.95\linewidth]{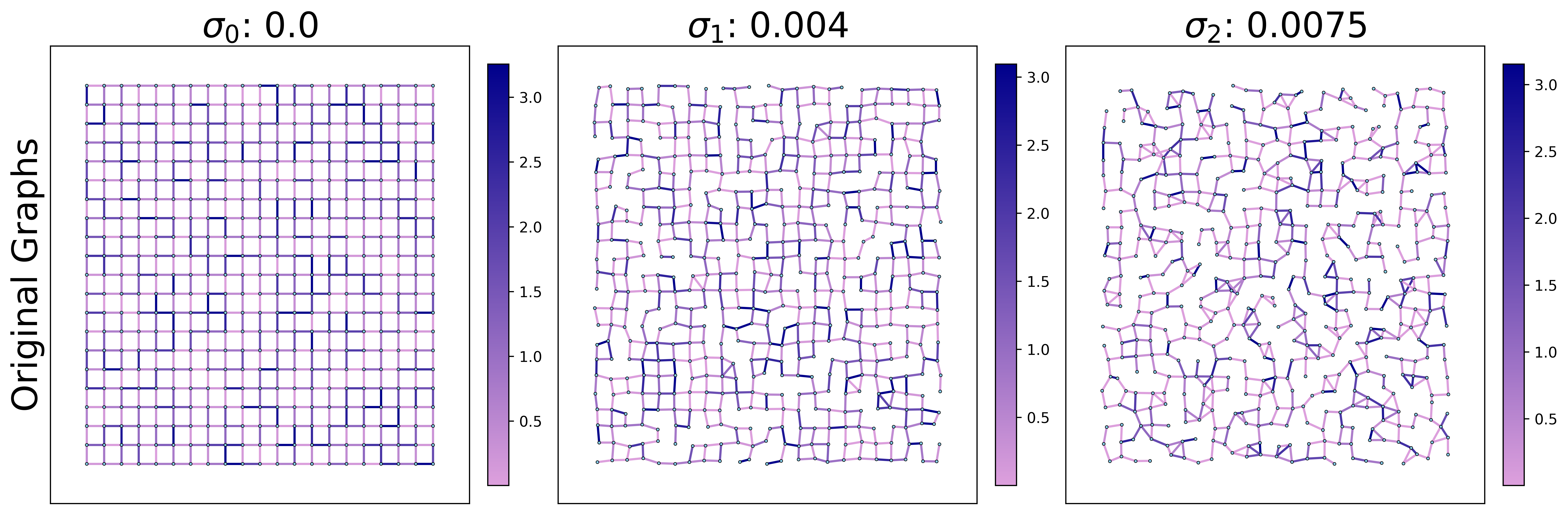}
    \caption{Random Geometric Graphs weighted with channel interference, with increasing positional noise added to reference grids.}
    \label{fig:parggsillustration}
\end{figure*}

For each PA and WLS, the optimization problem is solved using Lagrangian duality. The parameterized policies learn to maximize the Lagrangian and minimize the dual function. The details of the solutions and the learning algorithms are presented in \cite{eisen2020optimal, camargo2026longhorizonwirelesslinkscheduling}. The hyperparameter search consisted on finding the primal and dual learning rates, $\zeta_{\bbPhi}, \zeta_{\bblambda}$, respectively, with chosen parameters summarized in Table \ref{tab:hyperparams}. The primal learning rates are fixed, while the dual learning rates were adapted for each noise level. Given a noise level, the parameters were maintained across all scales, with the exception of $P_{\max}$ in PA, set to the number of devices in the network. 

\begin{table*}[]
    \centering
    \begin{tabular}{llllllll}
 &  & $\sigma_0$ & $\sigma_1$ & $\sigma_2$ & $\sigma_3$ & $\sigma_4$ & $\sigma_5$ \\ \hline
Power Allocation & $\zeta_{\bbPhi}$ & 1e-2 & 1e-2 & 1e-2 & 1e-2 & -- & -- \\ \cline{2-8}
 & $\zeta_{\bblambda}$ & 1e-3 & 1e-2 & 1e-2 & 1e-5 & -- & -- \\ \hline
Wireless Link Scheduling & $\zeta_{\bbPhi}$ & 5e-5 & 5e-5 & 5e-5 & 5e-5 & 5e-5 & 5e-5 \\ \cline{2-8} 
 & $\zeta_{\bblambda}$ & 5 & 2 & 2 & 5 & 2 & 5 \\ \cline{2-8} 
 & Weight decay & 0.05 & 0.06 & 0.05 & 0.065 & 0.065 & 0.05
\end{tabular}
    \caption{Hyperparameters for PA and WLS over different noise levels.}
    \label{tab:hyperparams}
\end{table*}

%% file: references.bib
@inproceedings{shen2017fplinq,
  title={{FPLinQ}: A cooperative spectrum sharing strategy for device-to-device communications},
  author={Shen, Kaiming and Yu, Wei},
  booktitle={2017 IEEE international symposium on information theory (ISIT)},
  pages={2323--2327},
  year={2017},
  organization={IEEE}
}

@article{levie2019transferability,
  title={Transferability of spectral graph convolutional neural networks},
  author={Levie, Ron and Huang, Wei and Bucci, Lorenzo and Bronstein, Michael and Kutyniok, Gitta},
  journal={Journal of Machine Learning Research},
  volume={22},
  number={272},
  pages={1--59},
  year={2021}
}

@article{keriven2020convergence,
  title={Convergence and stability of graph convolutional networks on large random graphs},
  author={Keriven, Nicolas and Bietti, Alberto and Vaiter, Samuel},
  journal={Advances in Neural Information Processing Systems},
  volume={33},
  pages={21512--21523},
  year={2020}
}

@article{gama2020stability,
  title={Stability properties of graph neural networks},
  author={Gama, Fernando and Bruna, Joan and Ribeiro, Alejandro},
  journal={IEEE Transactions on Signal Processing},
  volume={68},
  pages={5680--5695},
  year={2020},
  publisher={IEEE}
}

@article{maskey2021transferability,
  title={Transferability of Graph Neural Networks: an Extended Graphon Approach},
  author={Maskey, Sohir and Levie, Ron and Kutyniok, Gitta},
  journal={arXiv preprint arXiv:2109.10096},
  year={2021}
}

@article{ruiz2021transferability,
author = {Ruiz, Luana and Chamon, Luiz F. O. and Ribeiro, Alejandro},
title = {Transferability Properties of Graph Neural Networks},
year = {2023},
issue_date = {2023},
publisher = {IEEE Press},
volume = {71},
issn = {1053-587X},
url = {https://doi.org/10.1109/TSP.2023.3297848},
doi = {10.1109/TSP.2023.3297848},
journal = {Trans. Sig. Proc.},
month = jan,
pages = {3474–3489},
numpages = {16}
}

@inproceedings{wang2022stable,
  title={Stable and Transferable Wireless Resource Allocation Policies Via Manifold Neural Networks},
  author={Wang, Zhiyang and Ruiz, Luana and Eisen, Mark and Ribeiro, Alejandro},
  booktitle={ICASSP 2022-2022 IEEE International Conference on Acoustics, Speech and Signal Processing (ICASSP)},
  pages={8912--8916},
  year={2022},
  organization={IEEE}
}

@article{wang2024geometric,
  title={Geometric graph filters and neural networks: Limit properties and discriminability trade-offs},
  author={Wang, Zhiyang and Ruiz, Luana and Ribeiro, Alejandro},
  journal={IEEE Transactions on Signal Processing},
  year={2024},
  publisher={IEEE}
}

@article{eisen2020optimal,
  title={Optimal wireless resource allocation with random edge graph neural networks},
  author={Eisen, Mark and Ribeiro, Alejandro},
  journal={ieee transactions on signal processing},
  volume={68},
  month = {June},
  pages={2977--2991},
  year={2020},
  publisher={IEEE}
}

@ARTICLE{wang2022learningdecentralized,
  author={Wang, Zhiyang and Eisen, Mark and Ribeiro, Alejandro},
  journal={IEEE Transactions on Signal Processing}, 
  title={Learning Decentralized Wireless Resource Allocations With Graph Neural Networks}, 
  year={2022},
  volume={70},
  number={},
  pages={1850-1863},
  doi={10.1109/TSP.2022.3163626}}

@ARTICLE{linkschedulingusinggnns,
  author={Zhao, Zhongyuan and Verma, Gunjan and Rao, Chirag and Swami, Ananthram and Segarra, Santiago},
  journal={IEEE Transactions on Wireless Communications}, 
  title={Link Scheduling Using Graph Neural Networks}, 
  year={2023},
  volume={22},
  number={6},
  pages={3997-4012},
  doi={10.1109/TWC.2022.3222781}}

@article{tagconv,
  author       = {Jian Du and
                  Shanghang Zhang and
                  Guanhang Wu and
                  Jos{\'{e}} M. F. Moura and
                  Soummya Kar},
  title        = {Topology adaptive graph convolutional networks},
  journal      = {CoRR},
  volume       = {abs/1710.10370},
  year         = {2017},
  url          = {http://arxiv.org/abs/1710.10370},
  eprinttype    = {arXiv},
  eprint       = {1710.10370},
  bibsource    = {dblp computer science bibliography, https://dblp.org}
}

@misc{adam,
      title={Adam: A Method for Stochastic Optimization}, 
      author={Diederik P. Kingma and Jimmy Ba},
      year={2017},
      eprint={1412.6980},
      archivePrefix={arXiv},
      primaryClass={cs.LG},
      url={https://arxiv.org/abs/1412.6980}, 
}

@inproceedings{firstconflictgraph,
author = {Jain, Kamal and Padhye, Jitendra and Padmanabhan, Venkata N. and Qiu, Lili},
title = {Impact of interference on multi-hop wireless network performance},
year = {2003},
isbn = {1581137532},
publisher = {Association for Computing Machinery},
address = {New York, NY, USA},
url = {https://doi.org/10.1145/938985.938993},
doi = {10.1145/938985.938993},
booktitle = {Proceedings of the 9th Annual International Conference on Mobile Computing and Networking},
pages = {66–80},
numpages = {15},
location = {San Diego, CA, USA},
series = {MobiCom '03}
}

@misc{camargo2026longhorizonwirelesslinkscheduling,
      title={Long-Horizon Wireless Link Scheduling with State-Augmented Graph Neural Networks}, 
      author={Romina Garcia Camargo and Zhiyang Wang and Navid NaderiAlizadeh and Alejandro Ribeiro},
      year={2026},
      eprint={2607.18480},
      archivePrefix={arXiv},
      primaryClass={eess.SP},
      url={https://arxiv.org/abs/2607.18480}, 
}

@ARTICLE{hajekb88,
  author={Hajek, B. and Sasaki, G.},
  journal={IEEE Transactions on Information Theory}, 
  title={Link scheduling in polynomial time}, 
  year={1988},
  volume={34},
  number={5},
  pages={910-917},
  doi={10.1109/18.21215}}

@INPROCEEDINGS{garciacamargoscalabilitygnnrgg,
  author={Camargo, Romina Garcia and Wang, Zhiyang and Ribeiro, Alejandro},
  booktitle={ICASSP 2026 - 2026 IEEE International Conference on Acoustics, Speech and Signal Processing (ICASSP)}, 
  title={Graph Neural Networks in Large Scale Wireless Communication Networks: Scalability Across Random Geometric Graphs}, 
  year={2026},
  volume={},
  number={},
  pages={616-620},
  doi={10.1109/ICASSP55912.2026.11460522}}

@ARTICLE{gama19cnnforsignalsongraphs,
  author={Gama, Fernando and Marques, Antonio G. and Leus, Geert and Ribeiro, Alejandro},
  journal={IEEE Transactions on Signal Processing}, 
  title={Convolutional Neural Network Architectures for Signals Supported on Graphs}, 
  year={2019},
  volume={67},
  number={4},
  pages={1034-1049},
  doi={10.1109/TSP.2018.2887403}}

@INPROCEEDINGS{du18ongraphconvforgraphcnns,
  author={Du, Jian and Shi, John and Kar, Soummya and Moura, José M. F.},
  booktitle={2018 IEEE Data Science Workshop (DSW)}, 
  title={ON GRAPH CONVOLUTION FOR GRAPH CNNS}, 
  year={2018},
  volume={},
  number={},
  pages={1-5},
  doi={10.1109/DSW.2018.8439904}}

@ARTICLE{segarra17optimalgraphfilterdesign,
  author={Segarra, Santiago and Marques, Antonio G. and Ribeiro, Alejandro},
  journal={IEEE Transactions on Signal Processing}, 
  title={Optimal Graph-Filter Design and Applications to Distributed Linear Network Operators}, 
  year={2017},
  volume={65},
  number={15},
  pages={4117-4131},
  doi={10.1109/TSP.2017.2703660}}

@book{penrose2003random,
  title={Random geometric graphs},
  author={Penrose, Mathew},
  volume={5},
  year={2003},
  publisher={OUP Oxford}
}

@misc{owerko2023transferabilityconvolutionalneuralnetworks,
      title={Transferability of Convolutional Neural Networks in Stationary Learning Tasks}, 
      author={Damian Owerko and Charilaos I. Kanatsoulis and Jennifer Bondarchuk and Donald J. Bucci Jr and Alejandro Ribeiro},
      year={2023},
      eprint={2307.11588},
      archivePrefix={arXiv},
      primaryClass={cs.LG},
      url={https://arxiv.org/abs/2307.11588}, 
}

@ARTICLE{madala23cnnscurseofdimensionality,
  author={Madala, Vamshi C. and Chandrasekaran, Shivkumar and Bunk, Jason},
  journal={IEEE Open Journal of Signal Processing}, 
  title={CNNs Avoid the Curse of Dimensionality by Learning on Patches}, 
  year={2023},
  volume={4},
  number={},
  pages={233-241},
  doi={10.1109/OJSP.2023.3270082}}

@misc{camargo2026limitanalysisgraphneural,
      title={Limit Analysis of Graph Neural Networks with Wireless Conflict Graphs}, 
      author={Romina Garcia Camargo and Zhiyang Wang and Alejandro Ribeiro},
      year={2026},
      eprint={2606.03794},
      archivePrefix={arXiv},
      primaryClass={cs.LG},
      url={https://arxiv.org/abs/2606.03794}, 
}

@ARTICLE{wmmse,
  author={Christensen, Søren Skovgaard and Agarwal, Rajiv and De Carvalho, Elisabeth and Cioffi, John M.},
  journal={IEEE Transactions on Wireless Communications}, 
  title={Weighted sum-rate maximization using weighted MMSE for MIMO-BC beamforming design}, 
  year={2008},
  volume={7},
  number={12},
  pages={4792-4799},
  doi={10.1109/T-WC.2008.070851}}

@article{levin2026transferring,
  title={On transferring transferability: Towards a theory for size generalization},
  author={Levin, Eitan and Ma, Yuxin and D{\'\i}az, Mateo and Villar, Soledad},
  journal={Advances in Neural Information Processing Systems},
  volume={38},
  pages={67015--67083},
  year={2026}
}

@ARTICLE{zhou25wirelessdynamic,
  author={Zhou, Huan and Xia, Wenchao and Zheng, Gan and Zhu, Hongbo},
  journal={IEEE Transactions on Vehicular Technology}, 
  title={Graph Neural Network-Based Continual Learning for Resource Allocation in Dynamic Wireless Environments}, 
  year={2025},
  volume={74},
  number={12},
  pages={19125-19140},
  doi={10.1109/TVT.2025.3585146}}

@ARTICLE{uslu26faststateaugmented,
  author={Uslu, Yiğit Berkay and NaderiAlizadeh, Navid and Eisen, Mark and Ribeiro, Alejandro},
  journal={IEEE Transactions on Signal Processing}, 
  title={Fast State-Augmented Learning for Wireless Resource Allocation with Dual Variable Regression}, 
  year={2026},
  volume={},
  number={},
  pages={1-16},
  doi={10.1109/TSP.2026.3717245}}

@ARTICLE{shen2021scalablerrm,
  author={Shen, Yifei and Shi, Yuanming and Zhang, Jun and Letaief, Khaled B.},
  journal={IEEE Journal on Selected Areas in Communications}, 
  title={Graph Neural Networks for Scalable Radio Resource Management: Architecture Design and Theoretical Analysis}, 
  year={2021},
  volume={39},
  number={1},
  pages={101-115},
  doi={10.1109/JSAC.2020.3036965}}

@article{Wu2024OnTS,
  title={On the size generalizibility of graph neural networks for learning resource allocation},
  author={Jiajun Wu and Chengjian Sun and Chenyang Yang},
  journal={Science China Information Sciences},
  year={2024},
  volume={67},
  url={https://api.semanticscholar.org/CorpusID:268861972}
}

@inproceedings{
shu2026size,
title={Size Transferability of Graph Convolutional Networks across Sparsity: A Generalized Graphon Perspective},
author={Qinji Shu and Hang Sheng and Feng Ji and Hui Feng and Bo Hu},
booktitle={Forty-third International Conference on Machine Learning},
year={2026},
url={https://openreview.net/forum?id=V3qt2V78ly}
}

@inproceedings{le2023graphopssparse,
 author = {Le, Thien and Jegelka, Stefanie},
 booktitle = {Advances in Neural Information Processing Systems},
 doi = {10.52202/075280-1791},
 editor = {A. Oh and T. Naumann and A. Globerson and K. Saenko and M. Hardt and S. Levine},
 pages = {41305--41342},
 publisher = {Curran Associates, Inc.},
 title = {Limits, approximation and size transferability for GNNs on sparse graphs via graphops},
 url = {https://proceedings.neurips.cc/paper_files/paper/2023/file/8154c89c8d3612d39fd1ed6a20f4bab1-Paper-Conference.pdf},
 volume = {36},
 year = {2023}
}

@misc{shu2024transferabilitydownsampedsparsegraph,
      title={The Transferability of Downsamped Sparse Graph Convolutional Networks}, 
      author={Qinji Shu and Hang Sheng and Feng Ji and Hui Feng and Bo Hu},
      year={2024},
      eprint={2408.17274},
      archivePrefix={arXiv},
      primaryClass={cs.LG},
      url={https://arxiv.org/abs/2408.17274}, 
}

@ARTICLE{segarra2017optimalgraphfilter,
  author={Segarra, Santiago and Marques, Antonio G. and Ribeiro, Alejandro},
  journal={IEEE Transactions on Signal Processing}, 
  title={Optimal Graph-Filter Design and Applications to Distributed Linear Network Operators}, 
  year={2017},
  volume={65},
  number={15},
  pages={4117-4131},
  doi={10.1109/TSP.2017.2703660}}

@ARTICLE{sandryhalla2014dspongraphs,
  author={Sandryhaila, Aliaksei and Moura, José M. F.},
  journal={IEEE Transactions on Signal Processing}, 
  title={Discrete Signal Processing on Graphs: Frequency Analysis}, 
  year={2014},
  volume={62},
  number={12},
  pages={3042-3054},
  doi={10.1109/TSP.2014.2321121}}

@ARTICLE{perraudin2017stationary,
  author={Perraudin, Nathanaël and Vandergheynst, Pierre},
  journal={IEEE Transactions on Signal Processing}, 
  title={Stationary Signal Processing on Graphs}, 
  year={2017},
  volume={65},
  number={13},
  pages={3462-3477},
  doi={10.1109/TSP.2017.2690388}}

@ARTICLE{marques2017stationary,
  author={Marques, Antonio G. and Segarra, Santiago and Leus, Geert and Ribeiro, Alejandro},
  journal={IEEE Transactions on Signal Processing}, 
  title={Stationary Graph Processes and Spectral Estimation}, 
  year={2017},
  volume={65},
  number={22},
  pages={5911-5926},
  doi={10.1109/TSP.2017.2739099}}

@article{daviskahan,
 ISSN = {00361429},
 URL = {http://www.jstor.org/stable/2949580},
 author = {Chandler Davis and W. M. Kahan},
 journal = {SIAM Journal on Numerical Analysis},
 number = {1},
 pages = {1--46},
 publisher = {Society for Industrial and Applied Mathematics},
 title = {The Rotation of Eigenvectors by a Perturbation. III},
 urldate = {2026-08-25},
 volume = {7},
 year = {1970}
}
